%% file: root.tex
\documentclass[letterpaper, 10 pt, conference, onecolumn]{ieeeconf}  

\IEEEoverridecommandlockouts                              

\input{extrapackages}

\title{\LARGE \bf
Stability Bounds for Learning-Based Adaptive Control of Discrete-Time Multi-Dimensional Stochastic Linear Systems with Input Constraints
}

\author{Seth Siriya, Jingge Zhu, Dragan Ne\v{s}i\'{c}, and Ye Pu 
\thanks{All authors are with the Department of Electrical and Electronic Engineering, University of Melbourne, Parkville, 3010, Victoria, Australia. {\tt\small ssiriya@student.unimelb.edu.au}, {\tt\small\{jingge.zhu, dnesic, ye.pu\}@unimelb.edu.au}} 
\thanks{S. Siriya is supported by an Australian Government Research Training Program (RTP) Scholarship.}
}

\begin{document}

\maketitle
\thispagestyle{empty}
\pagestyle{empty}

\begin{abstract}
We consider the problem of adaptive stabilization for discrete-time, multi-dimensional linear systems with bounded control input constraints and unbounded stochastic disturbances, where the parameters of the true system are unknown. To address this challenge, we propose a certainty-equivalent control scheme which combines online parameter estimation with saturated linear control. We establish the existence of a high probability stability bound on the closed-loop system, under additional assumptions on the system and noise processes. Finally, numerical examples are presented to illustrate our results. 
\end{abstract}

\input{01-intro}

\input{02-problem}

\input{03-method}

\input{04-results}

\input{05-simulations}

\input{06-conclusion}




\bibliography{references.bib}
\bibliographystyle{ieeetr.bst}


\newpage

\section*{Supplementary Materials}
\input{08-supps.tex}


\addtolength{\textheight}{-12cm}

\end{document}

%% file: extrapackages.tex
\usepackage{algorithmic}
\let\labelindent\relax
\usepackage{enumitem}
\input{mathsetup}
\allowdisplaybreaks
\usepackage{xcolor}

\usepackage{caption}
\usepackage{subcaption}

%% file: mathsetup.tex
\usepackage{amssymb}

\usepackage{amsthm}
\usepackage{mathtools}
\usepackage{bm} 
\usepackage{graphicx} 
\usepackage{algorithm}
\usepackage{bbm} 
\usepackage{hyperref}
\usepackage{mathtools}
\usepackage{chngcntr}

\DeclareMathOperator{\supp}{supp}

\newcommand{\comp}{\mathsf{c}}

\newcommand{\abs}[1]{| #1 |}
\newcommand{\norm}[1]{\left\lVert #1 \right\rVert}

\newcommand{\trace}[1]{\text{tr} \left( #1 \right)}

\DeclareMathOperator{\tr}{tr}

\DeclareMathOperator{\sat}{sat}
\DeclareMathOperator{\proj}{proj}

\newtheorem{theorem}{Theorem}
\newtheorem{lemma}{Lemma}
\newtheorem{proposition}{Proposition}

\newtheorem{assumption}{Assumption}

\theoremstyle{definition}
\newtheorem{definition}{Definition}

\theoremstyle{remark}
\newtheorem{remark}{Remark}

\mathtoolsset{showonlyrefs=true} 

%% file: 01-intro.tex
\section{Introduction}

Adaptive control (AC) is concerned with the design of controllers for dynamical systems whose model parameters are unknown.
When deploying these algorithms in the real world, it is important that actuator saturation is accounted for during design, since ignorance of such issues may result in failure to achieve stability.
Moreover, systems are often subject to rare, large, disturbances --- often modelled by unbounded stochastic noise --- which can degrade control performance and potentially cause instability.
This motivates the need to develop provably stabilizing adaptive control algorithms that simultaneously handle input constraints and additive, unbounded, stochastic disturbances.

Discrete-time (DT) stochastic AC has recently seen renewed interest in the form of online model-based reinforcement learning --- especially for the online linear quadratic regulation (LQR) task, which aims to minimize regret with respect to the optimal LQR controller on an unknown, linear, stochastic system (see \cite{kargin2022thompson, simchowitz2020naive}). These results have been extended to handle state and input constraints \cite{li2021safe}, but only when disturbances are bounded. DT extremum seeking (ES) AC results have also shown promise for stabilizing unstable DT systems (\cite{ radenkovic2016stochastic, radenkovic2018extremum2}), but do not account for input constraints. Despite the long history of DT stochastic AC, going back to classic linear results such as \cite{goodwin1981discrete, guo1996self}, the control of non-strictly stable systems subject to input constraints has not garnered attention. Other nonlinear DT stochastic AC problems have been considered, such as dead-zone nonlinearities \cite{xiong1993stochastic}, and linearly parameterized nonlinear systems \cite{liu2022global}.
On the other hand, control constraints have been studied for the stabilization of unknown, DT output-feedback linear systems with bounded disturbances in \cite{chaoui2001adaptive,zhang2001adaptive}, but unbounded disturbances are not supported. Recently, mean square boundedness of a learning-based adaptive control scheme for at-worst marginally stable, scalar, linear systems, subject to Gaussian disturbances and bounded controls, was established in \cite{siriya2022learning}, by combining results from statistical learning theory \cite{simchowitz2018learning} with input-constrained stochastic control \cite{chatterjee2012mean}. However, multi-dimensional results are missing. 

Motivated by our previous scalar result \cite{siriya2022learning}, we move towards filling the gap in the multi-dimensional setting. In particular, we aim to develop a method for adaptive stabilization of unknown, multi-dimensional linear systems, subject to additive, i.i.d. sub-Gaussian zero-mean stochastic disturbances, and positive upper bound constraints on the control magnitude. Our main contributions are twofold:

Firstly, we propose a certainty-equivalence (CE) adaptive control scheme to address the problem. It consists of a saturated linear controller based on parameter estimates obtained via ordinary least squares (OLS) online, which has been intentionally excited by a bounded noise to facilitate parameter convergence. The saturation level and exciting noise level can be jointly selected to satisfy the control input constraint. 
Moreover, we do not assume prior knowledge of any bounds on the system parameters.

Secondly, we prove the existence of a high probability stability bound which holds on sub-sampled states of the closed-loop system, under the assumption that the system is controllable, $\Vert A\Vert \leq 1$, the saturation level of the CE component of the controller overpowers the statistics of the disturbance and exciting noise processes, and that a persistency of excitation-like condition holds on the state-input data sequence. To achieve this, we first establish an upper bound on the parameter estimation error that holds over time with high probability using tools from finite sample statistical learning theory \cite{simchowitz2018learning}. Then, we derive a probabilistic upper bound on the norm of the sub-sampled states which relies on a given estimation error bound. 
These two results are subsequently combined to derive a parameterized family of high probability upper bounds on the norm of the sub-sampled states. 
Our main result then follows.


\paragraph*{Notation} For a vector $x \in \mathbb{R}^n$, $| x |$ denotes its Euclidean norm. Given a matrix $M \in \mathbb{R}^{n\times m}$, $\Vert M\Vert $ is its induced $2$-norm, $\sigma_{\text{max}}(M)$ and $\sigma_{\text{min}}(M)$ denotes its maximum and minimum singular values respectively, $\mathcal{B}_r(M)$ denotes the $2$-norm open ball of radius $r > 0$ centered at $M$ and $\bar{\mathcal{B}}_r(M)$ denotes its closure, and $M^\dagger$ denotes its Moore-Penrose inverse. If $M\in \mathbb{R}^{n \times n}$ is symmetric, $\lambda_{\text{min}}(M)$ denotes its minimum eigenvalue, and $\lambda_{\text{max}}(M)$ denotes its maximum eigenvalue. 
For $r>0$, we define the saturation function $\text{sat}_r:\mathbb{R}^d\rightarrow\mathbb{R}^d$ by $\text{sat}_r(x) := x$ if $|x| \leq r$, and $\text{sat}_r(x) := r x / |x|$ if  $|x| > r$. 
The identity matrix is denoted by $\mathbf{I}$. 
The support of a function $f$ that maps from some set to $\mathbb{R}$ is denoted by $\supp(f)$. Given sets $A$ and $ B$, $A^{\comp}$ denotes the complement of $A$, $A\cap B$ denotes the intersection of $A$ and $B$, and $A \cup B$ denotes their union.
Consider a probability space $(\Omega,\mathcal{F}, P)$, and a random vector $X:\Omega \rightarrow \mathbb{R}^d$, an event $A \in \mathcal{F}$, and a sub-sigma-algebra $\mathcal{G}\subseteq \mathcal{F}$, defined on this space. 
The expected value value of $X$ is denoted by $\mathbb{E}[X]$. We define the indicator function $\bm{1}_A:\Omega \rightarrow \{0,1\}$ as $\bm{1}_A := 1$ on the event $A$, and $\bm{1}_A := 0$ on the event $A^{\comp}$. 
Denote the unit sphere embedded in $\mathbb{R}^d$ by $\mathcal{S}^{d-1}$. 
For scalar $X$, we say $X | \mathcal{G}$ is $\sigma^2$-sub-Gaussian if  $\mathbb{E}[e^{t X} | \mathcal{G}] \leq e^{\sigma^2 t^2 /2}$ for all $t \in \mathbb{R}$. For vector $X$, we say $X | \mathcal{G}$ is $\sigma^2$-sub-Gaussian if $\mathbb{E}[X]=0$, and $u^{\top}X | \mathcal{G}$ is $\sigma^2$-sub-Gaussian for any $u \in \mathcal{S}^{d-1}$.

%% file: 02-problem.tex
\section{Problem Setup} \label{sec:problem}

Consider the following stochastic linear system:
\begin{equation}
    X_{t+1} = AX_t + BU_t + W_t, \ t \in \mathbb{N}_0, \quad X_0 = x_0, \label{eqn:open-system}
\end{equation}
where the random sequences $(X_t)_{t \in \mathbb{N}_0}$, $(U_t)_{t \in \mathbb{N}_0}$ and $(W_t)_{t \in \mathbb{N}_0}$ are the states, controls, and disturbances, taking values in $\mathbb{R}^n$, $\mathbb{R}^m$, and $\mathbb{R}^n$ respectively, $x_0 \in \mathbb{R}^n$ is the initial state, and $A \in \mathbb{R}^{n \times n}$ and $B \in \mathbb{R}^{n \times m}$ are the true, unknown, system matrices. For convenience, we let $\theta_* = [A,B]$ denote the true system parameter. This is in contrast to $\hat{\theta}_t$, which denotes the estimated parameter at time $t$ and will be formally defined later. We make the following assumptions on the system in \eqref{eqn:open-system}. 

\begin{assumption} \label{assump:disturbance}
    The disturbance $(W_t)_{t \in \mathbb{N}_0}$ is an i.i.d. sequence that has an unbounded support, and is mean-zero and sub-Gaussian.
\end{assumption}

\begin{assumption} \label{assump:reachable}
    The matrix $A$ satisfies $\Vert A\Vert  \leq 1$, and $(A,B)$ is $\kappa$-step reachable with $\kappa \leq n$, that is, $\text{rank}(\mathcal{R}_{\kappa}(A,B)) = n$, where $\mathcal{R}_{\kappa}(A,B) := \begin{bmatrix} B & AB & \hdots & A^{\kappa-1}B  \end{bmatrix} \in \mathbb{R}^{n \times \kappa m}$. For ease of notation, we denote $\mathcal{R}_*=\mathcal{R}_{\kappa}(A,B)$.
\end{assumption}
\begin{remark}
Note that Assumption \ref{assump:disturbance} is broad enough to handle many different types of disturbance with an unbounded support, including Gaussian distributions. Since $W_t$ is assumed to be sub-Gaussian, its covariance matrix exists, which we denote by $\Sigma_W$. Assumption \ref{assump:reachable} is sufficient for guaranteeing the existence of control policies with bounded control constraints that render the system mean square bounded \cite{chatterjee2012mean}. This gives us hope that an adaptive control strategy is possible.
\end{remark}

Our goal is to formulate an adaptive control policy $(\pi_t)_{t \in \mathbb{N}_0}$ such that $\pi_t$ is a mapping from current and past state and control input data $(X_0,\hdots,X_t,U_0,\hdots,U_{t-1})$ and a randomizaton term $V_t$ to $\mathbb{R}^n$ for $t \in \mathbb{N}_0$. Here, $(V_t)_{t\in\mathbb{N}_0}$ taking values in $\mathbb{R}^{m}$ is an i.i.d. random sequence whose purpose is to excite the system in order to facilitate convergence of parameter estimates. The overall policy needs to be designed so that $|U_t| \leq U_{\text{max}}$ holds where $U_{\text{max}} > 0$ is the control magnitude constraint, whilst provably achieving stochastic stability guarantees on the closed-loop system states $(X_t)_{t \in \mathbb{N}}$ with $U_t = \pi_t(X_0,\hdots,X_t,U_0,\hdots,U_{t-1},V_t)$. Moreover, we require as part of our design that $\pi_t$ does not depend on the true system parameters $(A,B)$.


%% file: 03-method.tex
\section{Method and Main Result} \label{sec:method}


For the purpose of control design, we require knowledge of some $\kappa$ satisfying Assumption \ref{assump:reachable}. Although we can have $\kappa < n$ in many cases when systems have multiple inputs, if it is only known that $(A,B)$ is controllable, $\kappa = n$ is always a valid choice. Our control strategy is summarized in Algorithm \ref{alg:control-algorithm}.



\begin{algorithm}
\caption{Stochastic Adaptive Input-Constrained Control}
\begin{algorithmic}[1]
    \STATE \textbf{Inputs:} $U_{\text{max}} > 0$, $C \in (0,U_{\text{max}})$, $\bar{A}_0 \in \mathbb{R}^{n \times n}$, $\bar{B}_0 \in \mathbb{R}^{n \times m} \backslash \{0\} $, $\kappa \leq n$ \\
    \STATE $D \leftarrow U_{\text{max}} - C $
    \STATE $\tau \leftarrow 0$
    \STATE Measure $X_0$
    \FOR{$\tau = 0,1,\hdots$}
        \STATE Sample $V_{\kappa \tau},\hdots,V_{\kappa(\tau+1)-1}$ satisfying Assumption \ref{assump:exciting-noise}
        \STATE Compute $\bar{U}_{\tau}$ following \eqref{eqn:control-strategy}
        \FOR{$t = \kappa \tau, \hdots, \kappa (\tau + 1) - 1$}
            \STATE Extract $U_t$ from $\bar{U}_{\tau}$ following \eqref{eqn:control-strategy}
            \STATE Apply $U_t$ to \eqref{eqn:open-system}
            \STATE Measure $X_{t+1}$ from \eqref{eqn:open-system}
        \ENDFOR
        \STATE Compute $\bar{\theta}_{\tau+1}$ following \eqref{eqn:parameter-estimate}
    \ENDFOR
\end{algorithmic}\label{alg:control-algorithm}
\end{algorithm}

We now describe our strategy in greater detail. 
For all $\tau \in \mathbb{N}_0$, our sub-sampled control sequence $(\bar{U}_{\tau})_{\tau \in \mathbb{N}_0}$ is given by
\begin{align}
    \underbrace{\begin{bmatrix} U_{\kappa (\tau + 1) - 1 } \\ \vdots \\ U_{\kappa \tau} \end{bmatrix}}_{\bar{U}_{\tau}} = \text{sat}_{D}(-g(\bar{A}_{\tau},\bar{B}_{\tau}) X_{\kappa \tau} ) + \begin{bmatrix} V_{\kappa(\tau+1)-1} \\ \vdots \\ V_{\kappa \tau} \end{bmatrix} \label{eqn:control-strategy},
\end{align}
where $g(A',B') := \mathcal{R}_{\kappa}(A',B')^{\dagger} (A')^{\kappa}$, and $V_{t}$ is an additive \textit{excitation term} sampled so that Assumption \ref{assump:exciting-noise} is satisfied.


\begin{remark}
    When the true system parameter is used for control --- i.e. $\text{sat}_D(-g(A,B)x)$ is our control law --- the control policy can be described as a saturated deadbeat controller for the dynamical system obtained when \eqref{eqn:open-system} is sampled with periodicity $\kappa$. A similar controller structure was shown to achieve mean square boundedness in \cite{ramponi2010attaining}, except the saturation and and linear gain is switched. We opt for our order since our CE control strategy involves using estimates $(\bar{A}_{\tau},\bar{B}_{\tau})$ rather than $(A,B)$, and our estimates can be unbounded leading to unbounded gain. Applying saturation afterwards guarantees our controls satisfy $U_{\text{max}}$.
\end{remark}

\begin{assumption} \label{assump:exciting-noise}
    The random sequence $(V_t)_{t \in \mathbb{N}_0}$ taking values in $\mathbb{R}^m$ is i.i.d. Additionally, $| V_t | \leq C$ holds, and $V_i$,$W_j$ are independent for all $i,j \in \mathbb{N}_0$.
\end{assumption}
\begin{remark} 
Assumption \ref{assump:exciting-noise} restricts the magnitude of the additive noise $V_t$, which is required for satisfying control input constraints. We denote the covariance matrix for $V_t$ by $\Sigma_V$.
\end{remark}

Let $(\hat{\theta}_t)_{t \in \mathbb{N}}$ taking values in $\mathbb{R}^{n \times (n+m)}$ be the sequence of estimates of the true parameter $\theta_*$ at time obtained via OLS estimation:
\begin{equation}
    \hat{\theta}_{t} \in \arg \min_{\theta \in \mathbb{R}^{n \times (n+m)}} \sum_{s=1}^t \norm{X_s - \theta Z_s}_2^2, \label{eqn:parameter-estimate}
\end{equation}
where $(Z_t)_{t \in \mathbb{N}}$ taking values in $\mathbb{R}^{n + m}$ is the state-input data sequence, i.e. $Z_t = (X_{t-1},U_{t-1})$.
Let $(\bar{\theta}_{\tau})_{\tau \in \mathbb{N}}$ be the sequence of sub-sampled parameter estimates satisfying $\bar{\theta}_{\tau} = \hat{\theta}_{\kappa \tau}$, and let $(\bar{A}_{\tau})_{\tau \in \mathbb{N}},(\bar{B}_{\tau})_{\tau \in \mathbb{N}}$, be the sub-sampled estimates of $A$ and $B$ respectively, satisfying $[\bar{A}_{\tau},\bar{B}_{\tau}] = \bar{\theta}_{\tau}$. Note, the initial parameter estimate $(\bar{A}_0,\bar{B}_0)$ is not computed via OLS, but instead freely chosen by the designer in $\mathbb{R}^{n \times n}\times\mathbb{R}^{n \times m}\backslash \{0 \}$. Additionally, $C$ is a user-specified excitation constant satisfying $0 < C < U_{\text{max}}$ which determines the size of the excitation term, and $D = U_{\text{max}}-C$ is the magnitude of the certainty-equivalent component of the control policy.

Under this control strategy, the sub-sampled state sequence $(\bar{X}_{\tau})_{\tau \in \mathbb{N}_0}$, satisfying $\bar{X}_{\tau} = X_{\kappa \tau}$ evolves via the following closed loop system:
\begin{align} 
    \bar{X}_{\tau+1} &= AX_{\kappa(\tau+1)-1} + B U_{\kappa(\tau+1)-1} + W_{\kappa(\tau+1)-1}\\
    &=A^{\kappa}\bar{X}_{\tau} + \mathcal{R}_{\kappa}(A,B)\text{sat}_D(-g(\bar{A}_{\tau},\bar{B}_{\tau}) \bar{X}_{\tau}) + \bar{V}_{\tau} \\&\quad + \bar{W}_{\tau} \label{eqn:closed-loop-system}
\end{align}
for all $\tau \in \mathbb{N}_0$, where $\bar{V}_{\tau} = \mathcal{R}_* [V_{\kappa(\tau+1)-1}^{\top},\hdots, V_{\kappa \tau}^{\top}]^{\top}$, and $\bar{W}_\tau = \mathcal{R}_* [W_{\kappa(\tau+1)-1}^{\top},\hdots, W_{\tau}^{\top} ]^{\top}$.
Next, we let $M_{\bar{W}} = \ln( \mathbb{E}[e^{|\bar{W}_{\tau}|}] )$, and $M_{\bar{V}} = \ln( \mathbb{E}[e^{|\bar{V}_{\tau}|}] )$, whose existence are guaranteed from Assumptions \ref{assump:disturbance} and \ref{assump:exciting-noise}. We make the following assumption on their relationship with $D$ and $\mathcal{R}_*$.
\begin{assumption} \label{assump:large-sat-level}
    The saturation level $D$, $M_{\bar{V}}$, and $M_{\bar{W}}$, satisfy $\frac{D}{\Vert  \mathcal{R}_*^{\dagger} \Vert  } > M_{\bar{V}} + M_{\bar{W}}$.
\end{assumption}
\begin{remark}
    Assumption \ref{assump:large-sat-level} can be interpreted as saying that the magnitude of the certainty equivalent component of our controls is sufficiently large, such that it overpowers the statistics of the disturbance and the injected noise.
\end{remark}

Next, we define the block martingale small-ball (BMSB) condition, and assume that our state-input data sequence satisfies it.
\begin{definition} \label{def:bmsb}
(Martingale Small-Ball \cite[Definition 2.1]{simchowitz2018learning}) Given a process $(Z_t)_{t \geq 1}$ taking values in $\mathbb{R}^d$, we say that it satisfies the $(k,\Gamma_{\text{sb}},p)$-block martingale small-ball (BMSB) condition for $k \in \mathbb{N}$, $\Gamma_{\text{sb}} \succ 0$, and $p > 0$, if, for any $\zeta \in \mathcal{S}^{d-1}$ and $j \geq 0$, $\frac{1}{k} \sum_{i=1}^k P (\abs{\zeta^{\top}Z_{j+i}} \geq \sqrt{\zeta^{\top} \Gamma_{\text{sb}} \zeta} \mid \mathcal{F}_j) \geq p$ holds. Here, $(\mathcal{F}_t)_{t \geq 1}$ is any filtration which $(\zeta^{\top} Z_{t})_{t \geq 1}$ is adapted to.
\end{definition}

\begin{assumption} \label{assump:bmsb}
    The constants $k > 0$, $\Gamma_{\text{sb}} \succ 0$, and $p > 0$ are such that the state-input data sequence $(Z_t)_{t \in \mathbb{N}}$ satisfies the $(k,\Gamma_{\text{sb}},p)$-BMSB condition.
\end{assumption}

\begin{remark}
    The BMSB condition in Definition \ref{def:bmsb} can be used to establish that persistency of excitation holds, which is important for deriving high probability bounds on the estimation error (see \cite{tsiamis2022statistical}). By supposing $(Z_t)_{t \in \mathbb{N}}$ satisfies the BMSB condition, in Assumption \ref{assump:bmsb}, we are saying that conditioned on past $Z_t$, the averaged distributions of future $Z_t$ are sufficiently spread. We proved that it holds in the scalar case \cite{siriya2022learning}, and leave the vector case to future work. 
\end{remark}

We now present the main result of this paper, on the existence of a high probability stability bound for our learning-based adaptive control scheme.
\begin{theorem} \label{theorem:simplified-bound-v3}
    Suppose Assumptions \ref{assump:disturbance}, \ref{assump:reachable}, \ref{assump:exciting-noise}, \ref{assump:large-sat-level}, and \ref{assump:bmsb} hold. There exist $\lambda \in (0,1)$, $N_1,N_3>0$, and $N_2:\mathbb{R}^n \rightarrow \mathbb{R}_{\geq 0}$ such that
    \begin{align}
        |\bar{X}_{\tau}| &\leq \ln(N_2(x_0)(2/\delta)^{N_1}\lambda^{\tau} + N_3) \\
        & \quad + \ln(2/\delta) \label{eqn:theorem-simplified-bound}
    \end{align}
    holds with probability at least $1-\delta$ for all $x_0 \in \mathbb{R}^n$, $\delta \in (0,1)$ and $\tau \in \mathbb{N}_0$.
\end{theorem}
Theorem \ref{theorem:simplified-bound-v3} says that, for any initial state $x_0 \in \mathbb{R}^n$ and sub-sampled time $\tau \in \mathbb{N}_0$, with probability at least $1-\delta$, $\bar{X}_{\tau}$ will be in a ball around the origin with size equal to the right hand side of \eqref{eqn:theorem-simplified-bound}. In particular, we can interpret the result as a stability bound since the right hand side is uniformly bounded by $\ln(N_2(x_0)(2/\delta)^{N_1}+N_3) + \ln(2/\delta) $ over all $\tau \in \mathbb{N}_0$, and will asymptotically converge to $\ln(N_3) + \ln(2/\delta)$ regardless of $x_0$. Although the structure of this bound is non-standard, it can show up when bounding systems which converge to a set at a linear rate. 


%% file: 04-results.tex
\section{Proof of Main Result} \label{sec:results}

In this section, we build towards the proof of Theorem \ref{theorem:simplified-bound-v3}. In Section \ref{sec:estimation-error-bounds}, we establish a high probability upper bound on the parameter estimation error in the form of Proposition \ref{prop:estim-bound-all-time}. In Section \ref{sec:stability-bounds}, we provide in Proposition \ref{prop:stability-random-varying-param} a probabilistic upper bound on the norm of the sub-sampled states which relies on a given estimation error bound. This result is subsequently combined with Proposition \ref{prop:estim-bound-all-time} to derive a family of high probability stability bounds in Theorem \ref{theorem:stability-bound-all-epsilon}. Theorem \ref{theorem:simplified-bound-v3} then follows as a consequence. Although we provide sketches of the key ideas for proving Proposition \ref{prop:estim-bound-all-time}, Lemma \ref{lemma:control-error-simplified}, Lemma \ref{lemma:stability-deterministic-param}, and Lemma \ref{lemma:bound-time-stable}, we defer the formal proofs to the supplementary materials.

\subsection{Estimation Error Bound} \label{sec:estimation-error-bounds}
We first provide Proposition \ref{prop:estim-bound} from \cite[Theorem~2.4]{simchowitz2018learning}. It gives a high probability upper bounds on the estimation error for parameter estimates obtained by applying OLS to a general time-series with linear responses, and can successfully be applied when the BMSB condition in Definition~\ref{def:bmsb} is satisfied, and high-probability upper bounds on $\sum_{t=1}^T Z_t Z_t^{\top}$ can be found. 
\begin{proposition} \label{prop:estim-bound}
\cite[Theorem 2.4]{simchowitz2018learning} Consider some matrix $\theta_* \in \mathbb{R}^{n \times d}$. Fix $T \in \mathbb{N}$, $\delta \in (0,1)$, and $0 \prec \Gamma_{\text{sb}} \preceq \overline{\Gamma}$. Suppose $(Z_t,Y_t)_{t =1}^T \in ( \mathbb{R}^d \times \mathbb{R}^n )^T$ is a random sequence such that (a) $Y_t = \theta_* Z_t + \eta_t$ for $t \leq T$, where $\eta_t \mid \mathcal{F}_{t-1}$ is mean-zero and $\sigma^2$-sub-Gaussian with $\mathcal{F}_t$ denoting the sigma-algebra generated by $\eta_0, \hdots, \eta_t, Z_1, \hdots, Z_t$, (b) $Z_1,\hdots,Z_T$ satisfies the $(k, \Gamma_{\text{sb}}, p)$-BMSB condition, and (c) $P(\sum_{t=1}^T Z_t Z_t^{\top} \not \preceq T \overline{\Gamma}) \leq \delta$ holds. Then if $T \geq \frac{10 k}{p^2}\Big(\ln\Big(\frac{1}{\delta}\Big)+2d\ln(10/p)+\ln \det (\overline{\Gamma}\Gamma_{\text{sb}}^{-1})\Big)$, we have
\begin{align}
    &P\Bigg(\norm{\hat{\theta}_T-\theta_*}_{2} > \frac{90 \sigma}{p} \\
    &\times \sqrt{\frac{ n + d \ln \frac{10}{p} + \ln \det \overline{\Gamma} \Gamma_{\text{sb}}^{-1} + \ln \big(\frac{1}{\delta}\big) }{T \lambda_{\text{min}}(\Gamma_{\text{sb}}) }}\Bigg) \leq 3 \delta, \label{eqn:estim-bound-conclusion}
\end{align}
where $\hat{\theta}_T = (\bm{Z}_i^{\top} \bm{Z}_T)^{\dagger}\bm{Z}_T^{\top} \bm{Y}_T \in \arg \min_{\theta \in \mathbb{R}^{2}} \sum_{t=1}^T \Vert Y_t - \theta^{\top} Z_t\Vert^2$, $\bm{Z}_T=[Z_1,\hdots,Z_T]^{\top}$, $\bm{Y}_T = [Y_1,\hdots,Y_T]^{\top} $.
\end{proposition}

Under Assumptions \ref{assump:disturbance}, \ref{assump:reachable}, \ref{assump:exciting-noise}, and \ref{assump:bmsb}, we provide Proposition \ref{prop:estim-bound-all-time}  --- a high probability error bound on our parameter estimates from \eqref{eqn:parameter-estimate}. In particular, it says that, with probability at least $1-\delta$, the function $e(T,\delta)$ will bound the estimation error $\Vert \hat{\theta}_T - \theta_* \Vert$ over \textit{all} $T$ greater than $T_0(\delta,x_0)$. Alongside being specific to the parameter estimates in our problem, the key difference between this result and Proposition \ref{prop:estim-bound} is that the bound in Proposition~\ref{prop:estim-bound} holds with high probability for a specific $T$.
\begin{proposition} \label{prop:estim-bound-all-time}
Suppose Assumptions \ref{assump:disturbance}, \ref{assump:reachable}, \ref{assump:exciting-noise}, and \ref{assump:bmsb} hold. Consider the sequence of parameter estimates $(\hat{\theta}_t)_{t \in \mathbb{N}}$ from \eqref{eqn:parameter-estimate}. Fix $\delta \in (0,1)$. Then,
\begin{align}
    &P(\Vert \hat{\theta}_T-\theta_*\Vert _{2} \leq e(T,\delta) \text{ for all } T \geq T_0(\delta,x_0)) \geq 1 - \delta,
 \end{align}
where
\begin{align}
    &e(T,\delta,x_0) := \frac{90 \sqrt{\lambda_{\text{max}}(\Sigma_W)}}{p} \big ( (T \lambda_{\text{min}}(\Gamma_{\text{sb}}) )^{-1}  ( n \\
    & \ + (n + m) \ln \frac{10}{p} + \ln \det ( \frac{3(-1 + \pi^2/6) (T+1)^2}{\delta} \\
    & \ \times (4 |x_0|^2 + 2(D^2 + \trace{\Sigma_V}) +4(\Vert B\Vert^2 (D^2 + \trace{\Sigma_V}) \\
    & \ +  \trace{\Sigma_W} ) T^2 + \lambda_{\text{max}}(\Gamma_{\text{sb}})) \Gamma_{\text{sb}}^{-1}) \\
    & \ + \ln \big(\frac{3(-1 + \pi^2/6)(T+1)^2}{\delta}\big) ) \big )^{1/2}, \label{def:e} \\
    &T_0(\delta,x_0) \\
    &:= \min \{T_0' \in \mathbb{N} \mid T \geq \frac{10 k}{p^2}\Big( 2(n+m)\ln(10/p) \\
    & \ +\ln \det (\frac{3(-1 + \pi^2/6) (T+1)^2}{\delta} \times (4 |x_0|^2 + 2(D^2 \\
    & \ + \trace{\Sigma_V}) +4(\Vert B\Vert^2 (D^2 + \trace{\Sigma_V})  +  \trace{\Sigma_W} ) T^2 + \\
    & \ \lambda_{\text{max}}(\Gamma_{\text{sb}})) \Gamma_{\text{sb}}^{-1}) +\ln\Big(\frac{3(-1 + \pi^2/6)(T+1)^2}{\delta}\Big) \Big) \\
    & \ \text{ for all } T \geq T_0' \}. \label{def:T_0}
\end{align}
\end{proposition}
We defer the full proof of Proposition \ref{prop:estim-bound-all-time} to the supplementary materials, but describe the key ideas here. We first establish that the conditions for Proposition \ref{prop:estim-bound} hold by treating $Z_t \leftarrow (X_{t-1},U_{t-1})$ as the covariates and $Y_t \leftarrow X_t$ as the response. The BMSB condition holds by Assumption \ref{assump:bmsb}, and $P(\sum_{t=1}^T Z_t Z_t^{\top} \not \preceq T \overline{\Gamma}) \leq \delta$ is established with $\overline{\Gamma} \leftarrow 4 |x_0|^2 + 2(D^2  + \trace{\Sigma_V}) +4(\Vert B\Vert^2 (D^2 + \trace{\Sigma_V})  +  \trace{\Sigma_W} ) T^2 + \lambda_{\text{max}}(\Gamma_{\text{sb}})$. For our choice of $\delta$ in Proposition \ref{prop:estim-bound}, we set $\delta \leftarrow \delta/(3(-1+\pi^2/6)(T+1)^2)$, where the second $\delta$ is from the premise of Proposition \ref{prop:estim-bound-all-time}. The intuition here is we want to obtain a bound that holds with probability $1-\delta$ over all time, but Proposition \ref{prop:estim-bound} holds for a particular time. The conclusion follows by using the union bound to convert from a result at particular time to a result over all time, making use of $3\sum_{T\geq1}\delta/(3(-1 + \pi^2/6)(T+1)^2) =\delta$.

\subsection{Stability Bound} \label{sec:stability-bounds}

We state Lemma \ref{lemma:control-error-simplified}, which provides perturbation bounds on the CE component of the controls as a function of parameter estimation error and saturation level $D$, and holds uniformly over all states in the state space.
\begin{lemma} \label{lemma:control-error-simplified}
    Suppose Assumption \ref{assump:reachable} holds. There exist $m_q, M_q > 0$ such that for all $D > 0$ and $\epsilon \in [0, m_q]$,
    \begin{align}
        &|\sat_D(-g(\bar{A},\bar{B})x) - \sat_D(-g(A,B)x)| \leq M_q \cdot D \cdot \epsilon
    \end{align}
    holds for all $\bar{A} \in \bar{\mathcal{B}}_{\epsilon}(A)$, $\bar{B} \in \bar{\mathcal{B}}_{\epsilon}(B)$, and $x \in \mathbb{R}^n$.
\end{lemma}
We defer the proof of Lemma \ref{lemma:control-error-simplified} to the supplementary materials, but describe the key ideas here. This result follows from a perturbation bound that we derive on the controller saturation error using matrix analysis, which is convex in $\epsilon$ over a half-open interval. Of particular note in the proof, is that the perturbation bound holds uniformly over all states. This is because after fixing $D$ and $\epsilon$ sufficiently small, we can find a compact set $S$ of $x$ such that on $S^{\comp}$, both $\sat_D(-g(\bar{A},\bar{B})x)$ and $\sat_D(-g(A,B)x)$ are saturated, so the CE error will not grow with $|x|$. Although the CE error will grow with $|x|$ within $S$, we obtain a uniform bound based on the worst $|x|$.

Next, we provide Proposition \ref{prop:geometric-stability}, which is a slight modification of the geometric drift result in \cite[Proposition~1]{chatterjee2014stability}, and is used to derive stochastic stability results in expectation on Lyapunov-like functions of Markov processes.
\begin{proposition} \label{prop:geometric-stability}
Let $(\xi_t)_{t \in \mathbb{N}_0}$ be a Markov process taking values in $\mathbb{R}^d$, with $\xi_0$ having some distribution $\mu(\cdot)$. Suppose that there exist $\beta > 0$ and $\lambda >0 $, a measurable function $V: \mathbb{R}^d \rightarrow [0,+\infty)$, and a compact set $K \subset \mathbb{R}^d$, such that
\begin{equation}
    \mathbb{E}\left[ V(\xi_{ 1}) | \xi_{0} = x \right] \leq \lambda V(x) \text{ for all }  x \in K^{\comp} \cap \supp(\mu)   \label{eqn:geometric-condition-1}
\end{equation}
and
\begin{equation}
    \mathbb{E}\left[V(\xi_{1}) | \xi_0 = x \right] \leq \beta \text{ for all } x \in K \cap \supp(\mu) .  \label{eqn:geometric-condition-2}
\end{equation}
Then,
\begin{equation}
    \mathbb{E}\left[ V(\xi_t) \mid \xi_0 = x \right]  \leq \lambda^t V(x) + \beta \sum_{k=0}^{t-1} \lambda^{t-1-k}
\end{equation}
holds for all $x \in \supp(\mu)$, and $t \in \mathbb{N}_0$.
\end{proposition}
Proposition \ref{prop:geometric-stability} is the same as \cite[Proposition~1]{chatterjee2014stability}, except that we do not require $ \lambda < 1 $, and the final additive term in the upper bound on $\mathbb{E}\left[ V(\xi_t) \mid \xi_0 = x \right]$ is not replaced by the closed form of the geometric series.

We make use of Proposition \ref{prop:geometric-stability} to derive Lemma \ref{lemma:stability-deterministic-param}, which provides bounds on the expected value of an exponential function of the states of a parameterized family of systems that evolves with the same plant dynamics \eqref{eqn:closed-loop-system}, but with deterministic parameter estimates. These bounds hold uniformly over all estimates in a sufficiently small ball around the true parameter.
\begin{lemma} \label{lemma:stability-deterministic-param}
Suppose Assumptions \ref{assump:disturbance}, \ref{assump:reachable}, and \ref{assump:exciting-noise} hold. Let $m_q, M_q > 0$ satisfy Lemma \ref{lemma:control-error-simplified}. Fix $\epsilon \in [0,m_q]$. Consider a family of random sequences $(Z^{\bar{A},\bar{B}}_{\tau})_{\tau \in \mathbb{N}_0}$ parameterized by $\bar{A} \in \bar{\mathcal{B}}_{\epsilon}(A), \bar{B} \in \bar{\mathcal{B}}_{\epsilon}(B)$ that evolve according to the closed-loop system:
\begin{align}
    Z^{\bar{A},\bar{B}}_{\tau+1} = &A^{\kappa}Z^{\bar{A},\bar{B}}_{\tau} + \mathcal{R}_*\text{sat}_D(-\mathcal{R}_{\kappa}(\bar{A},\bar{B})^{\dagger} \bar{A}^{\kappa} Z^{\bar{A},\bar{B}}_{\tau}) \\
    &+ \tilde{V}_{\tau} + \tilde{W}_{\tau}, \label{eqn:parameterized-closed-loop-system}
\end{align}
for $\tau \in \mathbb{N}_0$, where:
\begin{enumerate}
    \item $(\tilde{W}_{\tau})_{\tau \in \mathbb{N}_0}$ and $(\tilde{V}_{\tau})_{\tau \in \mathbb{N}_0}$ are i.i.d. random sequences with the same distribution as $(\bar{W})_{\tau \in \mathbb{N}_0}$ and $(\bar{V}_{\tau})_{\tau \in \mathbb{N}_0}$ respectively;
    \item $Z^{\bar{A},\bar{B}}_0 = \underbar{Z}$ for all $\bar{A} \in \bar{\mathcal{B}}_{\epsilon}(A)$ and $\bar{B} \in \bar{\mathcal{B}}_{\epsilon}(B)$, where $\underbar{Z}$ is a random variable with distribution $\mu(\cdot)$;
    \item $(\tilde{W}_{\tau})_{\tau \in \mathbb{N}_0}$, $(\tilde{V}_{\tau})_{\tau \in \mathbb{N}_0}$, $\underbar{Z}$ are all independent.
\end{enumerate}
Then,
\begin{align}
    &\mathbb{E}[e^{|Z^{\bar{A},\bar{B}}_{\tau}|} \mid Z^{\bar{A},\bar{B}}_0 = z]  \leq \quad \lambda^{\tau}(\epsilon) e^{|z|} + \beta(\epsilon) \sum_{k = 0}^{\tau - 1}\lambda^{\tau - 1 - k}(\epsilon) \label{eqn:deterministic-param-bounds}
\end{align}
for all $\bar{A} \in \bar{\mathcal{B}}_{\epsilon}(A)$, $\bar{B} \in \bar{\mathcal{B}}_{\epsilon}(B)$, $z \in \supp(\mu)$ and $\tau \in \mathbb{N}_0$, where
\begin{align}
    &\lambda(\epsilon) := e^{\Vert \mathcal{R}_*\Vert \cdot M_q \cdot D \cdot \epsilon + \frac{-D}{\Vert  \mathcal{R}_*^{\dagger} \Vert } +  M_{\bar{V}} + M_{\bar{W}} }, \label{def:lambda}\\
    & \beta(\epsilon) := e^{\Vert \mathcal{R}_*\Vert \cdot M_q \cdot D \cdot \epsilon + M_{\bar{V}} + M_{\bar{W}}}. \label{def:beta}
\end{align}
\end{lemma}
We defer the proof of Lemma \ref{lemma:stability-deterministic-param} to the supplementary materials, but describe the key details here. By assuming $\bar{A} \in \bar{\mathcal{B}}_{\epsilon}(A)$ and $\bar{B} \in \bar{\mathcal{B}}_{\epsilon}(B)$, we can upper bound $|\bar{Z}_{1}^{\bar{A},\bar{B}}|$ based on the dynamics of the nominal sub-sampled closed-loop system (where the true parameter $A,B$ is used for the controller), and the CE error --- which, can be upper bounded using Lemma \ref{lemma:control-error-simplified}. This implies that $\mathbb{E}[e^{|Z^{\bar{A},\bar{B}}_1|} \mid Z_0 = z] \leq \mathbb{E}[e^{ | A^{\kappa}z + \mathcal{R}_*\text{sat}_D(-\mathcal{R}_*^{\dagger} A^{\kappa} z)| + M_{\bar{V}} + M_{\bar{W}} + \Vert \mathcal{R}_* \Vert M_q D \epsilon } \mid Z_0 = z ]$ holds. Letting $K$ denote the set of $z$ where $|\mathcal{R}_*^{\dagger} A^{\kappa} z)| \leq D$, we find that $\mathbb{E}[e^{|Z^{\bar{A},\bar{B}}_1|} \mid Z_0 = z] \leq \lambda(\epsilon)$ on $z \in K^{\comp}$, and $\mathbb{E}[e^{|Z^{\bar{A},\bar{B}}_1|} \mid Z_0 = z] \leq \beta(\epsilon)$ on $z \in K$. The conclusion follows by applying Proposition \ref{prop:geometric-stability} with $V(\cdot)\leftarrow e^{(\cdot)}$.
\begin{remark} \label{remark:lemma-stability-deterministic-param}
    The bound in \eqref{eqn:deterministic-param-bounds} always provides upper bounds on the conditional expectation of $e^{|Z^{\bar{A},\bar{B}}_{\tau}|}$ for $\tau > \tau_0$, and varies continuously as a function of the estimation error $\epsilon$. However, if $\epsilon  > 0$ is sufficiently small such that $\frac{-D}{\Vert  \mathcal{R}_*^{\dagger} \Vert } + \Vert \mathcal{R}_*\Vert M_qD \epsilon  + M_{\bar{V}} + M_{\bar{W}} < 0$ holds, then $\lambda(\epsilon) \in (0,1)$ will hold (under Assumption \ref{assump:large-sat-level}).
    In this scenario, the upper bound in \eqref{eqn:deterministic-param-bounds} will asymptotically converge towards a constant, and can be qualitatively interpreted as a stability bound.
\end{remark}

Using Lemma \ref{lemma:stability-deterministic-param}, we derive a probabilistic bound on the norm of the sub-sampled states $\bar{X}_{\tau}$ of the closed-loop system in \eqref{eqn:closed-loop-system} in Proposition \ref{prop:stability-random-varying-param}. Note that the upper bound here depends on the function $h$, which is an arbitrarily chosen function over $\tau$ representing an upper bound on the estimation error for the sub-sampled parameter estimates $\bar{\theta}_{\tau}$ that holds after time step $\tau_0$. The probability that this bound holds depends on the probability that the sub-sampled estimation error is bounded by $h$ between $\tau$ and $\tau_0$.
\begin{proposition} \label{prop:stability-random-varying-param}
Suppose Assumptions \ref{assump:disturbance}, \ref{assump:reachable}, and \ref{assump:exciting-noise} hold, and let $m_q,M_q > 0$ satisfy Lemma \ref{lemma:control-error-simplified}. Suppose $x_0 \in \mathbb{R}^n$. Consider any function $h: \{\tau_0, \tau_0 + 1, \hdots, \tau \} \rightarrow [0,m_q]$ with $\tau_0 \in \mathbb{N}$. Fix $\tau > \tau_0$ and $\gamma \in (0,1)$. Then,
\begin{align}
    |\bar{X}_{\tau}| <& \ln ( \frac{1}{\gamma} ) + \ln (  \mathbb{E}[ e^{|\bar{X}_{\tau_0}|} ] \prod_{i = \tau_0}^{\tau-1}\lambda(h(i)) \\
    &+ \sum_{i = \tau_0}^{\tau - 1} \beta(h(i)) \prod_{j=i+1}^{\tau-1} \lambda(h(j)) )  .
\end{align}
holds with at least probability $1 - \gamma - P(\bigcup_{i = \tau_0}^{\tau} \{ \bar{\theta}_i \not \in \bar{\mathcal{B}}_{h(i)} (\theta_*) \})$, with $\lambda,\beta$ defined in \eqref{def:lambda} and \eqref{def:beta} respectively.
\end{proposition}

\begin{proof}
    Using induction, we will first establish the following:
    \begin{align}
        &\mathbb{E}[\bm{1}_{\bar{\theta}_{\tau_0} \in \bar{\mathcal{B}}_{h(\tau_0)}(\theta_*),\hdots,\bar{\theta}_{\tau_0 + k - 1} \in \bar{\mathcal{B}}_{h(\tau_0 + k - 1)}(\theta_*)} e^{| \bar{X}_{\tau_0 + k} |} ] \\
        &\leq \mathbb{E}[e^{|\bar{X}_{\tau_0}|}] \prod_{i=\tau_0}^{\tau_0 + k - 1}\lambda(h(i)) + \sum_{i = \tau_0}^{\tau_0 + k - 1} \beta(h(i)) \prod_{j = i + 1}^{\tau_0 + k - 1} \lambda(h(j)) \label{eqn:prop-stability-random-varying-param-5}
    \end{align}
    Starting with the base case where $k = 1$, we have,
    \begin{align}
        &\mathbb{E}[ \bm{1}_{\bar{\theta}_{\tau_0} \in \bar{\mathcal{B}}_{h(\tau_0)}(\theta_*)} \mathbb{E} [ e^{| \bar{X}_{\tau_0 + 1} |} \mid \bar{\theta}_{\tau_0}, \bar{X}_{\tau_0} ] ] \label{eqn:prop-stability-random-varying-param-1} \\
        &\leq \mathbb{E}[ \bm{1}_{\bar{\theta}_{\tau_0} \in \bar{\mathcal{B}}_{h(\tau_0)}(\theta_*)} (\lambda(h(\tau_0)) e^{|\bar{X}_{\tau_0}|} + \beta(h(\tau_0)) ) ] \label{eqn:prop-stability-random-varying-param-2}
    \end{align}
    In order to see that \eqref{eqn:prop-stability-random-varying-param-2} holds, we define the parameterized sequence $(Z^{\bar{A},\bar{B}}_{i})_{i \in \mathbb{N}_0}$ over $\bar{A}\in \bar{\mathcal{B}}_{\delta}(A)$, $\bar{B}\in \bar{\mathcal{B}}_{\delta}(B)$, such that it evolves consistently with the closed-loop dynamics \eqref{eqn:parameterized-closed-loop-system}, but with $Z^{\bar{A},\bar{B}}_0 = \bar{X}_{\tau_0}$, and $(\tilde{W}_{i})_{i \in \mathbb{N}_0}$ and $(\tilde{V}_{i})_{i \in \mathbb{N}_0}$ satisfying $\tilde{W}_i = \bar{W}_{\tau_0 + i}$ and $\tilde{V}_i = \bar{V}_{\tau_0 + i}$ for $i \in \mathbb{N}_0$. Recall that $\bar{\theta}_{\tau_0} = [\bar{A}_{\tau_0}, \bar{B}_{\tau_0}].$ Letting $f(\bar{A},\bar{B},z) = \mathbb{E}[e^{|Z^{\bar{A},\bar{B}}_{1}|} \mid Z_0 = z]$ and noting that $\mathbb{E}[e^{|\bar{X}_{\tau_0+1}|} \mid \bar{\theta}_{\tau_0},\bar{X}_{\tau_0}] = \mathbb{E}[e^{|\bar{X}_{\tau_0+1}|} \mid \bar{A}_{\tau_0},\bar{B}_{\tau_0},\bar{X}_{\tau_0}] = f(\bar{A}_{\tau_0},\bar{B}_{\tau_0},\bar{X}_{\tau_0})$, it follows from Lemma \ref{lemma:stability-deterministic-param} that on the event $\{\bar{\theta}_{\tau_0} \in \bar{\mathcal{B}}_{h(\tau_0)}(\theta_*) \}$, $\mathbb{E}[e^{|\bar{X}_{\tau_0 + 1}|} \mid \bar{\theta}_{\tau_0},\bar{X}_{\tau_0}] \leq \lambda(h(\tau_0)) e^{|\bar{X}_{\tau_0}|} + \beta(h(\tau_0))$ holds. Thus, we have verified the base case.
    
    Now we move onto the inductive step. Suppose that $k \geq 1$ and \eqref{eqn:prop-stability-random-varying-param-5} holds for this $k$. Then, we have,
    \begin{align}
        &\mathbb{E}[\bm{1}_{\bar{\theta}_{\tau_0} \in \bar{\mathcal{B}}_{h(\tau_0)}(\theta_*),\hdots,\bar{\theta}_{\tau_0 + k } \in \bar{\mathcal{B}}_{h(\tau_0 + k)}(\theta_*)} e^{| \bar{X}_{\tau_0 + k + 1} |} ]  \\
        &= \mathbb{E}[ \bm{1}_{\bar{\theta}_{\tau_0} \in \bar{\mathcal{B}}_{h(\tau_0)}(\theta_*),\hdots,\bar{\theta}_{\tau_0 + k } \in \bar{\mathcal{B}}_{h(\tau_0 + k)}(\theta_*)} \\
        & \quad \times \mathbb{E} [ e^{| \bar{X}_{\tau_0 + k + 1} |} \mid \bar{\theta}_{\tau_0},\hdots,\bar{\theta}_{\tau_0 + k}, \bar{X}_{\tau_0 + k} ] ] \label{eqn:prop-stability-random-varying-param-7} \\
        &= \mathbb{E}[ \bm{1}_{\bar{\theta}_{\tau_0} \in \bar{\mathcal{B}}_{h(\tau_0)}(\theta_*),\hdots,\bar{\theta}_{\tau_0 + k } \in \bar{\mathcal{B}}_{h(\tau_0 + k)}(\theta_*)} \\
        &\quad \times \mathbb{E} [ e^{| \bar{X}_{\tau_0 + k + 1} |} \mid \bar{\theta}_{\tau_0 + k}, \bar{X}_{\tau_0 + k} ] ] \label{eqn:prop-stability-random-varying-param-8} \\
        &\leq \mathbb{E}[ \bm{1}_{\bar{\theta}_{\tau_0} \in \bar{\mathcal{B}}_{h(\tau_0)}(\theta_*),\hdots,\bar{\theta}_{\tau_0 + k } \in \bar{\mathcal{B}}_{h(\tau_0 + k)}(\theta_*)} \\
        &\quad \times ( \lambda(h(\tau_0 + k))e^{| \bar{X}_{\tau_0 + k} |} + \beta(h(\tau_0 + k)) ) ] \label{eqn:prop-stability-random-varying-param-9} \\
        &\leq \lambda(h(\tau_0 + k)) \mathbb{E}[ \bm{1}_{\bar{\theta}_{\tau_0} \in \bar{\mathcal{B}}_{h(\tau_0)}(\theta_*),\hdots,\bar{\theta}_{\tau_0 + k - 1 } \in \bar{\mathcal{B}}_{h(\tau_0 + k - 1)}(\theta_*)} \\[0em]
        &\quad \times e^{| \bar{X}_{\tau_0} + k |} ] + \beta(h(\tau_0 + k)) \label{eqn:prop-stability-random-varying-param-10} \\
        &\leq \lambda(h(\tau_0 + k)) \Big [ \mathbb{E}[e^{|\bar{X}_{\tau_0}|}] \prod_{i=\tau_0}^{\tau_0 + k - 1}\lambda(h(i)) \\
        &\quad + \sum_{i = \tau_0}^{\tau_0 + k - 1} \beta(h(i)) \prod_{j = i + 1}^{\tau_0 + k - 1} \lambda(h(j)) \Big ] + \beta(h(\tau_0 + k)) \label{eqn:prop-stability-random-varying-param-11} \\
        &= \mathbb{E}[e^{|\bar{X}_{\tau_0}|}] \prod_{i=\tau_0}^{\tau_0 + k}\lambda(h(i)) + \sum_{i = \tau_0}^{\tau_0 + k} \beta(h(i)) \prod_{j = i + 1}^{\tau_0 + k} \lambda(h(j)), \label{eqn:prop-stability-random-varying-param-12}
    \end{align}
    where \eqref{eqn:prop-stability-random-varying-param-7} follows from the tower property, \eqref{eqn:prop-stability-random-varying-param-8} follows from the conditional independence of $\bar{X}_{\tau_0 + k + 1}$ and $\bar{\theta}_{\tau_0},\hdots,\bar{\theta}_{\tau_0 + k-1}$ given $\bar{\theta}_{\tau_0 + k}, \bar{X}_{\tau_0 + k}$, \eqref{eqn:prop-stability-random-varying-param-9} follows via Lemma \ref{lemma:stability-deterministic-param} in a similar manner to \eqref{eqn:prop-stability-random-varying-param-2}, \eqref{eqn:prop-stability-random-varying-param-10} follows from $\bm{1}_{\bar{\theta}_{\tau_0 + k } \in \bar{\mathcal{B}}_{h(\tau_0 + k)}(\theta_*)} \leq 1$ and $\bm{1}_{\bar{\theta}_{\tau_0} \in \bar{\mathcal{B}}_{h(\tau_0)}(\theta_*),\hdots,\bar{\theta}_{\tau_0 + k - 1 } \in \bar{\mathcal{B}}_{h(\tau_0 + k - 1)}(\theta_*)} \leq 1$, and \eqref{eqn:prop-stability-random-varying-param-11} follows from the assumption that \eqref{eqn:prop-stability-random-varying-param-5} holds, which simplifies to \eqref{eqn:prop-stability-random-varying-param-12}. 
    Thus, we conclude via induction that \eqref{eqn:prop-stability-random-varying-param-5} holds for all $k \geq 1$.

    Letting $k = \tau - \tau_0$ and using \eqref{eqn:prop-stability-random-varying-param-5}, it follows that
    \begin{align}
        &\mathbb{E}[\bm{1}_{\bar{\theta}_{\tau_0} \in \bar{\mathcal{B}}_{h(\tau_0)}(\theta_*),\hdots,\bar{\theta}_{\tau - 1}(\theta_*) \in \bar{\mathcal{B}}_{h(\tau - 1)}} e^{| \bar{X}_{\tau} |} ] \\
        &\leq \mathbb{E}[e^{|\bar{X}_{\tau_0}|}] \prod_{i=\tau_0}^{\tau - 1}\lambda(h(i)) + \sum_{i = \tau_0}^{\tau - 1} \beta(h(i)) \prod_{j = i + 1}^{\tau - 1} \lambda(h(j)). \label{eqn:prop-stability-random-varying-param-6}
    \end{align}

    Now, define the events $\mathcal{E}^2 = \{ \bar{\theta}_{\tau_0} \in \bar{\mathcal{B}}_{h(\tau_0)}(\theta_*),\hdots,\bar{\theta}_{\tau - 1} \in \bar{\mathcal{B}}_{h(\tau - 1)}(\theta_*) \}$, and $\mathcal{E}^1_{\gamma} = \{ \bm{1}_{\mathcal{E}^2} e^{|\bar{X}_{\tau}|} < \frac{1}{\gamma} ( \mathbb{E}[e^{|\bar{X}_{\tau_0}|}] \prod_{i=\tau_0}^{\tau - 1}\lambda(h(i)) + \sum_{i = \tau_0}^{\tau - 1} \beta(h(i)) \prod_{j = i + 1}^{\tau - 1} \lambda(h(j)) ) \}$.
    From Markov's inequality and \eqref{eqn:prop-stability-random-varying-param-6}, it follows that $P(\mathcal{E}^1_{\gamma}) \geq 1 - \gamma$. On the event $\mathcal{E}^1_{\gamma} \cap \mathcal{E}^2$, $\bm{1}_{\mathcal{E}^2} e^{|\bar{X}_{\tau}|} = e^{|\bar{X}_{\tau}|}$ holds, which implies $e^{|\bar{X}_{\tau}|} < \frac{1}{\gamma}  ( \mathbb{E}[e^{|\bar{X}_{\tau_0}|}] \prod_{i=\tau_0}^{\tau - 1}\lambda(h(i)) + \sum_{i = \tau_0}^{\tau - 1} \beta(h(i)) \prod_{j = i + 1}^{\tau - 1} \lambda(h(j))  )$,
    which in turn is equivalent to $|\bar{X}_{\tau}| < \ln ( \frac{1}{\gamma} ) + \ln  ( \mathbb{E}[e^{|\bar{X}_{\tau_0}|}] \prod_{i=\tau_0}^{\tau - 1}\lambda(h(i)) + \sum_{i = \tau_0}^{\tau - 1} \beta(h(i)) \prod_{j = i + 1}^{\tau - 1} \lambda(h(j)) ) . $
    The conclusion follows by lower bounding the probability of this event:
    \begin{align}
        &P\big(|\bar{X}_{\tau}| < \ln ( \frac{1}{\gamma} ) + \ln \big( \big[ \mathbb{E}[e^{|\bar{X}_{\tau_0}|}] \prod_{i=\tau_0}^{\tau - 1}\lambda(h(i)) \\
        &+ \sum_{i = \tau_0}^{\tau - 1} \beta(h(i)) \prod_{j = i + 1}^{\tau - 1} \lambda(h(j)) \big] \big) \big) \\
        &\geq P(\mathcal{E}^1_{\gamma} \cap \mathcal{E}^2) \geq 1 - P((\mathcal{E}^1_{\gamma})^{\comp}) - P((\mathcal{E}^2)^{\comp}) \label{eqn:prop-stability-random-varying-param-15} \\
        &\geq 1 - \gamma - P(\bigcup_{k = \tau_0}^{\tau - 1} \{ \bar{\theta}_{k} \not \in \bar{\mathcal{B}}_{h(k)}(\theta_*) \}), \label{eqn:prop-stability-random-varying-param-16}
    \end{align}
    where \eqref{eqn:prop-stability-random-varying-param-15} follows from the union bound, and \eqref{eqn:prop-stability-random-varying-param-16} follows from $P(\mathcal{E}^1_{\gamma}) \geq 1 - \gamma$.
\end{proof}

Combining Propositions \ref{prop:estim-bound-all-time} and \ref{prop:stability-random-varying-param}, we are now ready to provide our main stability bound result in the form of Theorem \ref{theorem:stability-bound-all-epsilon}. It says, that given a failure probability $\delta$ and an estimation error parameter $\epsilon$, we have a corresponding high probability upper bound on the sub-sampled states if $\epsilon$ is sufficiently small. Note from Remark \ref{remark:lemma-stability-deterministic-param} that when $\epsilon \in (0, \min (m_q,(\Vert  \mathcal{R}_* \Vert M_q D)^{-1} ( \frac{D}{\Vert  \mathcal{R}_*^{\dagger} \Vert } - M_{\bar{V}} - M_{\bar{W}}  ) ))$ holds, $\lambda(\epsilon) \in (0,1)$ holds, and therefore the right hand side of \eqref{eqn:cor-stability-bound-all-epsilon-main-bound} is upper bounded by $\ln (2/\delta) + K(\epsilon,\delta/2,x_0) + \ln(\frac{\beta(\epsilon)}{1 - \lambda(\epsilon)})$ over all time, and asymptotically converges to $\ln (2/\delta) + \ln(\frac{\beta(\epsilon)}{1 - \lambda(\epsilon)})$ as $\tau \rightarrow \infty$. Because of this convergent behaviour, it can be viewed as providing a parameterized (in $\epsilon$) family of high probability stability bounds. 


\begin{theorem} \label{theorem:stability-bound-all-epsilon}
    Suppose Assumptions \ref{assump:disturbance}, \ref{assump:reachable}, \ref{assump:exciting-noise}, \ref{assump:large-sat-level}, and \ref{assump:bmsb} hold, and let $m_q,M_q>0$ satisfy Lemma \ref{lemma:control-error-simplified}. Suppose $x_0 \in \mathbb{R}^n$. Fix $\delta \in (0,1)$ and $\epsilon \in (0, \min (m_q,(\Vert  \mathcal{R}_* \Vert M_q D)^{-1} ( \frac{D}{\Vert  \mathcal{R}_*^{\dagger} \Vert } - M_{\bar{V}} - M_{\bar{W}}  ) ))$. Then, 
    \begin{align}
        |\bar{X}_{\tau}| < \ln(2/\delta) + \ln(K(\epsilon,\delta/2,x_0)\lambda^{\tau}(\epsilon) + \frac{\beta(\epsilon)}{1 - \lambda(\epsilon)}) \label{eqn:cor-stability-bound-all-epsilon-main-bound}
    \end{align}
    holds with probability at least $1 - \delta$ for all $\tau \geq 0$, where
    \begin{align}
        K(\epsilon,\delta,x_0) &:=  e^{|x_0|} ( e^{\Vert \mathcal{R}_*\Vert D + M_{\bar{V}} + M_{\bar{W}}} \lambda^{-1})^{ \tau_0'(\epsilon,\delta,x_0) }, \\
        \tau_0'(\epsilon,\delta,x_0) &:= \min \{ \tau \in \mathbb{N} \mid \kappa \tau \geq T_0(\delta,x_0),  \\
        &\quad e(\kappa i,\delta, x_0) \leq \epsilon \text{ for all } i \geq \tau \}, \label{def:tau_0'} 
    \end{align}
    with $T_0$, $e$, $\lambda$, $\beta$ defined in \eqref{def:T_0}, \eqref{def:e}, \eqref{def:lambda}, \eqref{def:beta} respectively.
\end{theorem}

\begin{proof}
    We begin by deriving a worst-case high probability bound on the states. For all $\tau \geq 0$, we have
    \begin{align}
        |\bar{X}_{\tau}| &= |A^{\kappa \tau} \bar{X}_0 + \sum_{i=0}^{\tau-1} A^{\kappa (\tau - 1 - i)} (\mathcal{R}_* \sat_D(-\mathcal{R}_{\kappa}(\bar{A}_i,\bar{B}_i)\bar{X}_{i}) \\
        &\quad + \bar{V}_i + \bar{W}_i)  | \label{eqn:theorem-stability-final-1} \\
        &\leq |x_0| + \sum_{i=0}^{\tau-1}  ( \Vert  \mathcal{R}_* \Vert D  + |\bar{V}_i| + |\bar{W}_i|), \label{eqn:theorem-stability-final-2}
    \end{align}
    where \eqref{eqn:theorem-stability-final-1} follows via iterative application of \eqref{eqn:closed-loop-system}, and \eqref{eqn:theorem-stability-final-2} holds from Assumption \ref{assump:reachable} and $\sat_D(\cdot) \leq D$. From the monotonicity of expectation and exponential functions, it follows that
    \begin{align}
        &\mathbb{E}[e^{|\bar{X}_{\tau}|}] \\
        &\leq e^{|x_0|} \prod_{i=0}^{\tau-1}(e^{ \Vert  \mathcal{R}_* \Vert D} \mathbb{E}[e^{|\bar{V}_i|}] \mathbb{E}[e^{|\bar{W}_i|}]) \label{eqn:theorem-stability-final-3} \\
        &\leq e^{ |x_0|} \prod_{i=0}^{\tau-1}(e^{ \Vert  \mathcal{R}_* \Vert D + M_{\bar{V}} + M_{\bar{W}}} ) \label{eqn:theorem-stability-final-4} \\
        &= e^{\Vert A\Vert^{\kappa \tau} | x_0 |} ( e^{\Vert \mathcal{R}_* \Vert D + M_{\bar{V}} + M_{\bar{W}}} )^{\sum_{i=0}^{\tau-1} \Vert  A\Vert^{\kappa i}}  \\
        &\leq e^{ | x_0 |} ( e^{\Vert \mathcal{R}_* \Vert D + M_{\bar{V}} + M_{\bar{W}}} )^{\tau} \label{eqn:theorem-stability-final-5}
    \end{align}
    holds for all $\tau \geq 0$, where \eqref{eqn:theorem-stability-final-3} holds due to Assumption \ref{assump:exciting-noise}. Using Markov's inequality, it follows from \eqref{eqn:theorem-stability-final-5} that for all $\gamma \in (0,1)$,
    \begin{align}
        e^{|\bar{X}_{\tau}|} < \frac{1}{\gamma}( e^{ | x_0 |} ( e^{\Vert \mathcal{R}_* \Vert D + M_{\bar{V}} + M_{\bar{W}}} )^{\tau} ) \label{eqn:combined-main-result-1}
    \end{align}
    holds with probability at least $1-\gamma$ for all $\tau \geq 0$. Equivalently, we have that for all $\gamma \in (0,1)$, 
    \begin{align}
        |\bar{X}_{\tau}| < \ln{(\frac{1}{\gamma} e^{ | x_0 |} ( e^{\Vert \mathcal{R}_* \Vert D + M_{\bar{V}} + M_{\bar{W}}} )^{\tau } )} \label{eqn:combined-main-result-2}
    \end{align}
    holds with probability at least $1-\gamma$ for all $\tau \geq 0$.

    We now proceed by considering the case where $\tau$ satisfies $0 \leq \tau \leq \tau_0'(\epsilon,\delta/2,x_0)$ and $\tau > \tau_0'(\epsilon,\delta/2,x_0)$ separately.

    \textit{Case 1:} Suppose $\tau > \tau_0'(\epsilon,\delta/2,x_0)$. Then,
    \begin{align}
        &P(\bigcup_{i = \tau_0'(\epsilon,\delta/2,x_0)}^{\tau} \{ \bar{\theta}_i \not \in \bar{\mathcal{B}}_{e(\kappa i, \delta/2, x_0)}(\theta_*) \} ) \\
        &\leq 1 - P(\bigcap_{i \geq \tau_0'(\epsilon,\delta/2,x_0)} \{ \bar{\theta}_i \in \bar{\mathcal{B}}_{e(\kappa i, \delta/2, x_0)}(\theta_*) \} ) \\
        &\leq 1 - P(\Vert \bar{\theta}_i-\theta_*\Vert _{2} \leq e(\kappa i, \delta/2, x_0) \\
        & \quad \quad \text{for all } i \geq \frac{1}{\kappa}T_0(\delta/2,x_0)) \label{eqn:theorem-stability-final-6} \\
        & \leq 1 - P(\Vert \hat{\theta}_{T}-\theta_*\Vert _{2} \leq e(T,\delta/2,x_0)\\
        &  \quad \quad  \text{for all } T \geq  T_0(\delta/2,x_0)) \label{eqn:theorem-stability-final-7} \\
        & \leq \delta/2 \label{eqn:theorem-stability-final-8}
    \end{align}
    where \eqref{eqn:theorem-stability-final-6} follows from the definition of $\tau_0'$ in \eqref{def:tau_0'}, and \eqref{eqn:theorem-stability-final-7} holds since $\bar{\theta}_{i} = \hat{\theta}_{\kappa i}$, alongside $P(\Vert \hat{\theta}_{T}-\theta_*\Vert _{2} \leq e(\kappa T, \delta,x_0) \text{ for all } T \geq  T_0(\delta,x_0)) \leq P(\Vert \hat{\theta}_{\kappa i}-\theta_*\Vert _{2} \leq e(\kappa i, \delta, x_0) \text{ for all } i \geq \frac{1}{\kappa}T_0(\delta,x_0))$. Moreover, \eqref{eqn:theorem-stability-final-8} holds following Proposition \ref{prop:estim-bound-all-time}. Combining \eqref{eqn:theorem-stability-final-8} with Proposition \ref{prop:stability-random-varying-param} and setting $h(i)\leftarrow e(\kappa i,\delta/2,x_0)$, $\tau_0 \leftarrow \tau_0'(\epsilon,\delta/2,x_0)$, and $\gamma \leftarrow \delta/2$, we find that
    \begin{align}
        &|\bar{X}_{\tau}| < \ln ( 2/\delta ) + \ln(  \mathbb{E}[ e^{|\bar{X}_{\tau_0}|} ] \prod_{i = \tau_0'(\epsilon,\delta/2,x_0)}^{\tau-1}\lambda(e(\kappa i, \delta/2, x_0)) \\
        & \ + \sum_{i = \tau_0'(\epsilon,\delta/2,x_0)}^{\tau - 1} \beta(e(\kappa i, \delta/2, x_0)) \prod_{j=i+1}^{\tau-1} \lambda(e(\kappa j, \delta/2, x_0)) ) \label{eqn:prop-stability-random-param-2-eq3} \\
        &\leq \ln ( 2/\delta ) + \ln( e^{ | x_0 |} ( e^{\Vert \mathcal{R}_* \Vert D + M_{\bar{V}} + M_{\bar{W}}} )^{\tau_0'(\epsilon,\delta/2,x_0)}  \lambda^{\tau}(\epsilon) \\
        & \ +   \beta(\epsilon)  \frac{1 - \lambda^{\tau- \tau_0'(\epsilon,\delta/2,x_0)}(\epsilon)}{1-\lambda(\epsilon)}  ) \label{eqn:combined-main-result-3}\\
        & \leq  \ln(2/\delta) + \ln(K(\epsilon,\delta/2)\lambda^{\tau}(\epsilon) + \frac{\beta(\epsilon)}{1 - \lambda(\epsilon)}) \label{eqn:combined-main-result-4}
    \end{align}
    holds with at least probability $1 - \delta$ for all $\tau > \tau_0'(\epsilon,\delta/2,x_0)$, where \eqref{eqn:combined-main-result-3} follows from the definition of $\tau_0'$ in \eqref{def:tau_0'} and by applying \eqref{eqn:combined-main-result-2} with the substitution $\gamma \leftarrow \delta/2$ at $\tau \leftarrow \tau_0'(\epsilon,\delta/2,x_0)$, and \eqref{eqn:combined-main-result-4} follows after simplification.
    
    \textit{Case 2:} Suppose $\tau \leq \tau_0'(\epsilon,\delta/2,x_0)$. Then, by applying \eqref{eqn:combined-main-result-2} with the substitution $\gamma \leftarrow \delta/2$, we find that for all $\tau \geq 0$,
    \begin{align}
        |\bar{X}_{\tau}| &< \ln(2/\delta) + \ln( e^{| x_0 |} ( e^{\Vert \mathcal{R}_* \Vert D + M_{\bar{V}} + M_{\bar{W}}} )^{\tau} ) \label{eqn:cor-stability-bound-all-epsilon-4}
    \end{align}
    holds with probability at least $1 - \delta/2$, implying it also holds with probability at least $1 - \delta$. Next, note that
    \begin{align}
        &e^{| x_0 |} ( e^{\Vert \mathcal{R}_* \Vert D + M_{\bar{V}} + M_{\bar{W}}} )^{\tau}  \\
        &\leq  e^{| x_0 |} ( e^{\Vert \mathcal{R}_* \Vert D + M_{\bar{V}} + M_{\bar{W}}} \lambda^{-1} )^{\tau_0'(\epsilon,\delta/2,x_0)}  \\
        &= K(\epsilon,\delta/2)\lambda^{\tau}(\epsilon) + \frac{\beta(\epsilon)}{1 - \lambda(\epsilon)} \label{eqn:cor-stability-bound-all-epsilon-5}
    \end{align}
    holds. By substituting \eqref{eqn:cor-stability-bound-all-epsilon-5} into \eqref{eqn:cor-stability-bound-all-epsilon-4}, it follows that \eqref{eqn:cor-stability-bound-all-epsilon-main-bound} has been verified for $\tau$ satisfying $0 \leq \tau \leq \tau_0'(\epsilon,\delta/2,x_0)$.
\end{proof}

Before proving Theorem 1, we provide Lemma \ref{lemma:bound-time-stable} to simplify $\tau_0'$ from \eqref{def:tau_0'} for a given $\epsilon$. 
\begin{lemma} \label{lemma:bound-time-stable}
    Suppose Assumptions \ref{assump:disturbance}, \ref{assump:reachable}, \ref{assump:exciting-noise}, and \ref{assump:bmsb} hold. For all $\epsilon > 0$, there exist $L_1 > 0$ and a function $L_2 : \mathbb{R}^n \rightarrow \mathbb{R}_{>0}$ such that $\tau_0'(\epsilon,\delta) \leq L_2(x_0) + L_1 \ln(1/\delta)$ for all $x_0 \in \mathbb{R}^n$ and $\delta \in (0,1)$.
\end{lemma}
We defer the proof of this result to the supplementary materials, but describe the key details here. This result follows by first upper bounding $T_0(\delta)$ by using the fact that $\log(T+1)$ can linearly upper bounded by its tangent since it is concave, then using simple but tedious algebraic manipulation to find a sufficient bound independent of $T$ such that $T$ greater than this bound implies $T \geq T_0(\delta)$. This is repeated for upper bounding the minimum time $T'$ that $e(t,\delta,x_0) \leq \epsilon$ for all $t \geq T'$. Finally, $\tau_0'(\delta)$ is bounded by simplifying the maximum of these two bounds, then dividing by $\kappa$ to translate to sub-sampled time.

Our main result --- Theorem \ref{theorem:simplified-bound-v3} --- immediately follows from Theorem \ref{theorem:stability-bound-all-epsilon} and Lemma \ref{lemma:bound-time-stable}.
\begin{proof}[Proof of Theorem \ref{theorem:simplified-bound-v3}]
    Let $m_q,M_q>0$ satisfy Lemma \ref{lemma:control-error-simplified}. Let $\epsilon$ satisfy $\epsilon\in(0, \min ( m_q, (\Vert  \mathcal{R}_* \Vert M_q D)^{-1} ( \frac{D}{\Vert  \mathcal{R}_*^{\dagger} \Vert } - M_{\bar{V}} - M_{\bar{W}}  ) ))$, $N_3 = \beta(\epsilon)/(1-\lambda(\epsilon))$, and $\lambda_{\circ} = \lambda(\epsilon)$. Let $L_1 > 0$, $L_2:\mathbb{R}^n \rightarrow \mathbb{R}_{\geq 0}$ satisfy Lemma \ref{lemma:bound-time-stable} for the chosen $\epsilon$, such that $\tau_0'(\epsilon,\delta,x_0) \leq L_2(x_0) + L_1 \ln(1/\delta)$ holds. Let $N_1 = (\Vert \mathcal{R}_*\Vert D + M_{\bar{V}} + M_{\bar{W}} + \ln (\lambda^{-1}(\epsilon)))L_1$, and $N_2(x_0) = e^{|x_0|} ( e^{(\Vert \mathcal{R}_*\Vert D + M_{\bar{V}} + M_{\bar{W}} + \ln (\lambda^{-1}(\epsilon))) L_2(x_0)} )$. Then, we have
    \begin{align}
        K(\epsilon,\frac{\delta}{2},x_0) &=  e^{|x_0|} ( e^{\Vert \mathcal{R}_*\Vert D + M_{\bar{V}} + M_{\bar{W}} + \ln (\lambda^{-1}(\epsilon))} )^{ \tau_0'\big(\epsilon,\frac{\delta}{2},x_0\big) } \\
        &\leq  N_2(x_0)(2/\delta)^{N_1}
    \end{align}
    The result then follows by applying Theorem \ref{theorem:stability-bound-all-epsilon}.
\end{proof}

%% file: 05-simulations.tex
\section{Numerical Examples} \label{sec:sim}

To demonstrate the effectiveness of our adaptive control strategy in Algorithm \ref{alg:control-algorithm}, we first tested it on three different plants where $W_t \stackrel{i.i.d.}{\sim} \mathcal{N}(0,\mathbf{I})$:
\begin{enumerate}
    \item $A_1 = \begin{bmatrix} \cos (\pi/4) & \sin (\pi/4) \\ -\sin (\pi/4) & \cos (\pi/4) \end{bmatrix}$, $B_1 = \begin{bmatrix} 0 \\ 1 \end{bmatrix}$;
    \item $A_2 = \begin{bmatrix} \cos (-\pi/2) & \sin (-\pi/2) \\ -\sin (-\pi/2) & \cos (-\pi/2) \end{bmatrix}$, $B_2 = \begin{bmatrix} 0.3 \\ -0.5 \end{bmatrix}$;
    \item $A_3 = \begin{bmatrix} 0.8 \cos (\pi/4) & 0.8 \sin (\pi/4) \\ -0.8\sin (\pi/4) & 0.8\cos (\pi/4) \end{bmatrix}$, $B_3 = \begin{bmatrix} 0.5 \\ 0 \end{bmatrix}$.
\end{enumerate}
The algorithm parameters were also fixed to $U_{\text{max}}\leftarrow 1$, $C \leftarrow 0.4$, $\kappa \leftarrow 2$, with $V_t \stackrel{\text{i.i.d.}}{\sim}\text{Uniform}(-C,C)$ and $(\bar{A}_0,\bar{B}_0)$ randomly selected (but fixed across all trials).
Moreover, we additionally simulated system $(A_1,B_1)$ when it is uncontrolled (i.e. $U_t=0$). The plots of the median and 90th percentiles of $|X_t|$ over 100 trials are shown in Fig. \ref{fig:sim-compare-sys-v2}. It can be seen that in all tests where Algorithm \ref{alg:control-algorithm} is applied, at both the median and 90th percentile, $|X_t|$ seems to exhibit stable behaviour in the sense of boundedness, which is consistent with Theorem \ref{theorem:simplified-bound-v3}. This is in contrast to the case with no controls, where the median and 90th percentile plots have unbounded growth. 

Secondly, we tested Algorithm \ref{alg:control-algorithm} using the same algorithm parameters on $(A_1,B_1)$ again, but with $\Sigma_W = 0.1 \mathbf{I}$ and varying $x_0$. The plots are shown in Figure \ref{fig:sim-compare-x0}. We see that regardless of the initial state, both the median and $90th$ percentile converge. In particular, if we focus individually on either the median or $90$th percentile plots, and vary $x_0$, it appears that convergence occurs to the same steady state. Moreover, this convergence \textit{seems} linear. This is consistent with the trends of the upper bound in Theorem \ref{theorem:simplified-bound-v3}.

\begin{figure}[h]
    \begin{subfigure}[b]{0.49\textwidth}
        \centering
        \includegraphics[width=0.6\textwidth]{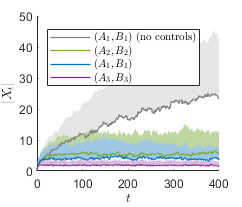}
        \caption{Simulation for varying $(A,B)$, and with no controls.}
        \label{fig:sim-compare-sys-v2}
    \end{subfigure}
    \hfill
    \begin{subfigure}[b]{0.49\textwidth}
        \centering
        \includegraphics[width=0.6\textwidth]{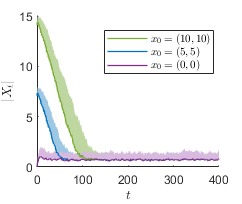}
        \caption{Simulation for varying $x_0$.}
        \label{fig:sim-compare-x0}
    \end{subfigure}%
    \caption{Plots for median and 90th percentile of $|X_t|$ over 100 trials.}
    \label{fig:median-90-plots}
\end{figure}


%% file: 06-conclusion.tex
\section{Conclusion} \label{sec:conclusion}
We proposed an excited CE control scheme for adaptive control of multi-dimensional, stochastic, linear systems subject to additive, i.i.d, unbounded stochastic disturbances, with positive upper bound constraints on the control magnitude. Moreover, we established a high probability stability bound on the $\kappa$-sub-sampled states of the closed-loop system. The stability of our control strategy is verified in numerical examples.

This work can be extended in several directions. Our method has the potential to be extended to linear systems where $(A,B)$ are stabilizable, $\lambda_{\text{max}}(A)=1$, and eigenvalues of $A$ on the unit circle have equal algebraic and geometric multiplicity, since the existence of mean square stabilizing controllers has been demonstrated \cite{ramponi2010attaining}. However, controller design based on \cite{ramponi2010attaining} involves a similarity transformation $T(A)$ that takes $A$ to real Jordan form. Algorithm \ref{alg:control-algorithm} could be modified to support such controllers, however, stability analysis would require inspection of the continuity properties of $T(A)$. Another interesting direction is the consideration of output-feedback problems, since we only address the full-state feedback setting. Overcoming Assumption \ref{assump:large-sat-level} is also of interest, since it is known that in the non-adaptive setting, arbitrarily small controls are sufficient for stochastic stability \cite{chatterjee2012mean}. We leave this to future work.


%% file: 08-supps.tex
\subsection{Perturbation Bounds for Certainty Equivalent Component of Controls}

\begin{proof}[Proof of Lemma \ref{lemma:control-error-simplified}]
    Let 
    \begin{align}
        m_q&=q_1(\kappa,A,B)/2, \label{def:m_q}\\
        M_q&= \frac{4 q_5(q_1(\kappa, A,B)/2,\kappa,A,B)}{ q_1(\kappa, A, B) (\sigma_{\text{min}}(g(A,B)) - q_5(q_1(\kappa, A,B)/2,\kappa,A,B))}, \label{def:M_q}
    \end{align}
    with $q_1, q_5$ defined in \eqref{def:q_1} and \eqref{def:q_5} respectively. Suppose $\bar{A} \in \bar{\mathcal{B}}_{\epsilon}(A)$, $\bar{B} \in \bar{\mathcal{B}}_{\epsilon}(B)$, and $x \in \mathbb{R}^n$. Then,
    \begin{align}
        &|\sat_D(-g(A_2,B_2)x) - \sat_D(-g(A,B)x)| \\
        & \leq q_2(q_1(\kappa,A,B)/2, D, \kappa,A,B) \label{eqn:lemma-control-error-simplified-1} \\
        & \leq \frac{2q_2(q_1(\kappa, A,B)/2, D, \kappa, A, B )}{q_1(\kappa, A, B)}  \epsilon \label{eqn:lemma-control-error-simplified-2} \\
        &= M_q(\kappa, A, B) D \epsilon, \label{eqn:lemma-control-error-simplified-3}
    \end{align}
    holds, where \eqref{eqn:lemma-control-error-simplified-1} follows from Lemma \ref{lemma:control-error} with $q_2$ defined in \ref{def:q_2}, \eqref{eqn:lemma-control-error-simplified-2} follows from the convexity of $q_2(\epsilon,D,\kappa,A,B)$ in $\epsilon$ over the interval $[0,q_1(\kappa,A,B)/2]$, and \eqref{eqn:lemma-control-error-simplified-3} follows \eqref{def:M_q}. The convexity of $q_2$ in $\epsilon$ holds since it is constructed via 1) the composition of convex functions where the deepest function is convex in $\epsilon$, and 2) the multiplication and addition of convex functions in $\epsilon$, over this interval.
\end{proof}

We state Lemma \ref{lemma:control-error}, which provides perturbation bounds on the certainty-equivalent component of the controls as a function of parameter estimation error, and holds uniformly over all states in the state space.
\begin{lemma} \label{lemma:control-error}
    Consider $(A_1,B_1) \in \mathbb{R}^{n \times n} \times \mathbb{R}^{n \times m}$ that is $\kappa$-step reachable, and $r >0$. Fix $\epsilon \in [0, q_1(\kappa,A_1,B_1) )$. Then,
    \begin{align}
        &|\sat_r(-g(A_2,B_2)x) - \sat_r(-g(A_1,B_1)x)| \\
        & \leq q_2(\epsilon,r,\kappa,A_1,B_1)
    \end{align}
    holds for all $A_2 \in \bar{\mathcal{B}}_{\epsilon}(A_1)$, $B_2 \in \bar{\mathcal{B}}_{\epsilon}(B_1)$, and $x \in \mathbb{R}^n$. The functions $q_1$ and $q_2$ are defined as
    \begin{align}
        &q_1(\kappa,A,B) := \min \{ q_3(\kappa,A,B), q_4(\kappa,A,B)\}, \label{def:q_1} \\
        &q_2(\epsilon,r,\kappa,A,B) := \frac{2 r q_5(\epsilon,\kappa,A,B)}{\sigma_{\text{min}}(g(A,B)) - q_5(\epsilon,\kappa,A,B)}. \label{def:q_2}
    \end{align}
    They are supported by $q_3, q_4, q_5, q_6, q_7$, and $q_8$, which are defined as
    \begin{align}
        &q_3(\kappa,A,B) := \sup \{ a > 0 \mid q_6(\epsilon,\kappa,A,B) < \\
        & \quad \sigma_{\text{min}} (\mathcal{R}_{\kappa}(A,B)\mathcal{R}^{\top}_{\kappa}(A,B) ) \text{ for all } \epsilon \in [0,a] \},\\
        &q_4(\kappa,A,B) := \sup \{ a > 0 \mid q_5(\epsilon,\kappa,A,B) < \\
        &\quad \sigma_{\text{min}} (g(A,B)) \text{ for all } \epsilon \in [0,a] \}, \\
        & q_5(\epsilon,\kappa,A,B):= q_7(\epsilon,\kappa,A,B) q_{10}(\epsilon,\kappa,A) \label{def:q_5} \\
        & \quad + \Vert  R^{\dagger}_{\kappa}(A,B) \Vert q_{10}(\epsilon,\kappa,A) + \Vert  A^{\kappa} \Vert q_7(\epsilon,\kappa,A,B), \\
        & q_6(\epsilon,\kappa,A,B):= ( \sum_{i=0}^{\kappa-1} q_9(\epsilon,i,A,B) )^2 + \\
        & \quad 2 \Vert  \mathcal{R}_{\kappa}(A,B) \Vert ( \sum_{i=0}^{\kappa-1} q_9(\epsilon,i,A,B)) \label{def:q_6} \\
        & q_7(\epsilon,\kappa,A,B):=(\sum_{i=0}^{\kappa-1} q_9(\epsilon,i,A,B)) q_8(\epsilon,\kappa,A,B) \\
        & \quad + \Vert \mathcal{R}_{\kappa}(A,B)\Vert q_8(\epsilon,\kappa,A,B) \\
        &\quad + \Vert (R^{\top}_{\kappa}(A,B)R_{\kappa}(A,B))^{-1}\Vert (\sum_{i=0}^{\kappa-1} q_9(\epsilon,i,A,B)) \label{def:q_7} \\
        & q_8(\epsilon,\kappa,A,B) :=  \\
        & \quad \frac{\Vert  ( \mathcal{R}_{\kappa}(A,B)\mathcal{R}_{\kappa}^{\top}(A,B) )^{-1} \Vert^2  q_6(\kappa,\epsilon,A,B)}{1 - \Vert  ( \mathcal{R}_{\kappa}(A,B)\mathcal{R}_{\kappa}^{\top}(A,B) )^{-1} \Vert q_6(\kappa,\epsilon,A,B)} \label{def:q_8},
    \end{align}
    with $q_9,q_{10}$ defined in \eqref{def:q_9} and \eqref{def:q_10} respectively.
\end{lemma}

\begin{proof}

For ease of notation, let $\mathcal{R}_1' = \mathcal{R}_{\kappa}(A_1,B_1)$, and $\mathcal{R}_2' = \mathcal{R}_{\kappa}(A_2,B_2)$. Suppose $\epsilon \in [0,q_1(\kappa,A_1,B_1))$, $A_2 \in \bar{\mathcal{B}}_{\epsilon}(A_1)$, $B_2 \in \bar{\mathcal{B}}_{\epsilon}(B_1)$, and $x \in \mathbb{R}^n$. Then,
\begin{align}
    &\Vert \mathcal{R}_2' - \mathcal{R}_1'\Vert \\
    &= \Vert \begin{bmatrix} 
        B_2-B_1 & A_2B_2 - A_1B_1 & \hdots & A_2^{\kappa-1}B_2 - A_1^{\kappa-1}B_1
    \end{bmatrix}\Vert \label{eqn:lemma-control-error-1} \\
    &\leq \sum_{i = 0}^{\kappa-1} \Vert  A_2^iB_2 - A_1^iB_1 \Vert \\
    &\leq \sum_{i = 0}^{\kappa-1} q_9(\epsilon,i,A_1,B_1), \label{eqn:lemma-control-error-2}
\end{align}
where \eqref{eqn:lemma-control-error-1} follows from the definition of $\mathcal{R}_{\kappa}$ in Assumption \ref{assump:reachable}, and \eqref{eqn:lemma-control-error-2} follows from Lemma \ref{lemma:perturb-matrix-power-product} with $q_9$ defined in \eqref{def:q_9}. 
Next, we have,
\begin{align}
    & \Vert  \mathcal{R}_2' (\mathcal{R}_2')^{\top} - \mathcal{R}_1' (\mathcal{R}_1')^{\top} \Vert \\
    & \leq \Vert \mathcal{R}_2' - \mathcal{R}_1'  \Vert^2 + 2\Vert \mathcal{R}_1' \Vert \Vert \mathcal{R}_2' - \mathcal{R}_1'  \Vert \label{eqn:lemma-control-error-3} \\
    &\leq q_6(\epsilon,\kappa,A_1,B_1), \label{eqn:lemma-control-error-4}
\end{align}
where \eqref{eqn:lemma-control-error-3} follows from Lemma \ref{lemma:perturb-bound-matrix-product}, and \eqref{eqn:lemma-control-error-4} follows from \eqref{eqn:lemma-control-error-2} and the definition of $q_6$ in \eqref{def:q_6}.
Then,
\begin{align}
    &\Vert  (\mathcal{R}_2' (\mathcal{R}_2')^{\top})^{-1} - (\mathcal{R}_1' (\mathcal{R}_1')^{\top})^{-1} \Vert \\
    &\leq \frac{\Vert  (\mathcal{R}_1' (\mathcal{R}_1')^{\top})^{-1} \Vert^2 \Vert  \mathcal{R}_2' (\mathcal{R}_2')^{\top} - \mathcal{R}_1' (\mathcal{R}_1')^{\top} \Vert}{1 - \Vert (\mathcal{R}_1' (\mathcal{R}_1')^{\top})^{-1} \Vert \Vert  \mathcal{R}_2' (\mathcal{R}_2')^{\top} - \mathcal{R}_1' (\mathcal{R}_1')^{\top} \Vert } \label{eqn:lemma-control-error-5} \\
    &\leq q_8(\epsilon,\kappa,A_1,B_1), \label{eqn:lemma-control-error-6}
\end{align}
where \eqref{eqn:lemma-control-error-5} holds via Lemma \ref{lemma:perturb-bound-matrix-inverse}, and the fact that $\Vert \mathcal{R}_2'(\mathcal{R}_2')^{\top} - \mathcal{R}_1'(\mathcal{R}_1')^{\top}\Vert < \sigma_{\text{min}}(\mathcal{R}_1'(\mathcal{R}_1')^{\top})$ is satisfied since $\epsilon < q_1(\kappa,A_1,B_1)$, as well as the definition of $q_1$ in \eqref{def:q_1}. Moreover, \eqref{eqn:lemma-control-error-6} follows from \eqref{eqn:lemma-control-error-4} and the definition of $q_8$ in \eqref{def:q_8}. Next, note that $\mathcal{R}_2'$ is full rank, and $\Vert \mathcal{R}_2'(\mathcal{R}_2')^{\top} - \mathcal{R}_1'(\mathcal{R}_1')^{\top}\Vert < \sigma_{\text{min}}(\mathcal{R}_1'(\mathcal{R}_1')^{\top})$ implies that $\mathcal{R}_1'(\mathcal{R}_1')^{\top}$ is full rank, and hence, so is $\mathcal{R}_1'$ since it is a fat matrix. Since $\mathcal{R}_1', \mathcal{R}_2'$ are both full rank and fat, it follows that $(\mathcal{R}_2')^{\dagger} = (\mathcal{R}_2')^{\top}(\mathcal{R}_2' (\mathcal{R}_2')^{\top})^{-1}$ and $(\mathcal{R}_1')^{\dagger} = (\mathcal{R}_1')^{\top}(\mathcal{R}_1' (\mathcal{R}_1')^{\top})^{-1}$. Thus, we have
\begin{align}
    &\Vert  (\mathcal{R}_2')^{\dagger} - (\mathcal{R}_1')^{\dagger} \Vert \\
    &= \Vert  (\mathcal{R}_2')^{\top}(\mathcal{R}_2' (\mathcal{R}_2')^{\top})^{-1} - (\mathcal{R}_1')^{\top}(\mathcal{R}_1' (\mathcal{R}_1')^{\top})^{-1} \Vert\\
    &\leq \Vert \mathcal{R}_2' - \mathcal{R}_1' \Vert \Vert  (\mathcal{R}_2' (\mathcal{R}_2')^{\top})^{-1} - (\mathcal{R}_1' (\mathcal{R}_1')^{\top})^{-1} \Vert \\
    & \quad + \Vert \mathcal{R}_1'\Vert \Vert (\mathcal{R}_2' (\mathcal{R}_2')^{\top})^{-1} - (\mathcal{R}_1' (\mathcal{R}_1')^{\top})^{-1} \Vert \\
    & \quad + \Vert (\mathcal{R}_1' (\mathcal{R}_1')^{\top})^{-1}\Vert \Vert \mathcal{R}_2' - \mathcal{R}_1' \Vert \label{eqn:lemma-control-error-7}\\
    &\leq q_7(\epsilon,\kappa,A_1,B_1) \label{eqn:lemma-control-error-8}
\end{align}
where \eqref{eqn:lemma-control-error-7} holds via Lemma \ref{lemma:perturb-bound-matrix-product}, and \eqref{eqn:lemma-control-error-8} is satisfied via \eqref{eqn:lemma-control-error-2}, \eqref{eqn:lemma-control-error-6}, and the definition of $q_7$ in \eqref{def:q_7}. Recalling the definition of $g$ from \eqref{eqn:control-strategy}, we have
\begin{align}
    &\Vert  g(A_2,B_2) - g(A_1,B_1) \Vert \\
    &= \Vert  (\mathcal{R}_2')^{\dagger} A_2^{\kappa} - (\mathcal{R}_1')^{\dagger} A_1^{\kappa} \Vert \\
    &\leq \Vert (\mathcal{R}_2')^{\dagger} - (\mathcal{R}_1')^{\dagger} \Vert \Vert  A_2^{\kappa} - A_1^{\kappa} \Vert \\
    & \quad + \Vert (\mathcal{R}_1')^{\dagger}\Vert \Vert A_2^{\kappa} - A_1^{\kappa} \Vert + \Vert A_1^{\kappa}\Vert \Vert (\mathcal{R}_2')^{\dagger} - (\mathcal{R}_1')^{\dagger} \Vert \label{eqn:lemma-control-error-9} \\
    &\leq q_5(\epsilon,\kappa,A_1,B_1) \label{eqn:lemma-control-error-10}
\end{align}
where \eqref{eqn:lemma-control-error-9} follows from Lemma \ref{lemma:perturb-bound-matrix-product}, and \eqref{eqn:lemma-control-error-10} follows from the fact $\Vert A_2^{\kappa} - A_1^{\kappa} \Vert \leq q_{10}(\epsilon,\kappa,A_1)$ using Lemma \ref{lemma:perturb-bound-matrix-power} with $q_{10}$ defined in \ref{def:q_10}, alongside \eqref{eqn:lemma-control-error-8} and the definition of $q_5$ in \eqref{def:q_5}.

Now, let $S = \{ x \in \mathbb{R}^n \mid |x| \leq \frac{r}{\sigma_{\text{min}}(g(A_1,B_1)) - q_5(\epsilon,\kappa,A_1,B_1)} \}.$ Our proof proceeds by considering the case where $x \in S$ and $x \not \in S$ separately.

\textit{Case 1:} Suppose $x \in S$. Then,
\begin{align}
    &| \sat_r(-g(A_2,B_2)x) - \sat_r(-g(A_1,B_1)x) | \\
    &\leq \Vert  g(A_2,B_2) -g(A_1,B_1) \Vert \\
    & \quad \times \frac{r}{\sigma_{\text{min}}(g(A_1,B_1)) - q_5(\epsilon,\kappa,A_1,B_1)} \label{eqn:lemma-control-error-11} \\
    &\leq  \frac{r q_5(\epsilon,\kappa,A_1,B_1)}{\sigma_{\text{min}}(g(A_1,B_1)) - q_5(\epsilon,\kappa,A_1,B_1)}, \label{eqn:lemma-control-error-12}
\end{align}
where \eqref{eqn:lemma-control-error-11} follows from Lemma \ref{lemma:sat-func-bound} and $x \in S$, and \eqref{eqn:lemma-control-error-12} follows from \eqref{eqn:lemma-control-error-10}.

\textit{Case 2:} Suppose $x \not \in S$. Then, $|g(A_1,B_1)x| \geq \sigma_{\text{min}}(g(A_1,B_1)) |x| \geq (\sigma_{\text{min}}(g(A_1,B_1))-q_5(\epsilon,\kappa,A_1,B_1)) |x| \geq r$ and $|g(A_2,B_2)x| \geq \sigma_{\text{min}}(g(A_2,B_2)) |x| \geq (\sigma_{\text{min}}(g(A_1,B_1))-q_5(\epsilon,\kappa,A_1,B_1)) |x| \geq r$, and therefore,
\begin{align}
    &| \sat_r(-g(A_2,B_2)x) - \sat_r(-g(A_1,B_1)x) | \\
    &\leq r | \frac{g(A_2,B_2)x}{|g(A_2,B_2)x|} - \frac{g(A_1,B_1)x}{|g(A_1,B_1)x|} | \label{eqn:lemma-control-error-13} \\
    & \leq r\frac{2 \Vert  g(A_2,B_2) - g(A_1,B_1) \Vert}{\sigma_{\text{min}}(g(A_1,B_1)) - \Vert  g(A_2,B_2) - g(A_1,B_1) \Vert } \label{eqn:lemma-control-error-14} \\
    & \leq \frac{2 r q_5(\epsilon,\kappa,A_1,B_1)}{\sigma_{\text{min}}(g(A_1,B_1)) - q_5(\epsilon,\kappa,A_1,B_1)}, \label{eqn:lemma-control-error-15}
\end{align}
where \eqref{eqn:lemma-control-error-13} follows from Lemma \ref{lemma:sat-func-bound}, \eqref{eqn:lemma-control-error-14} follows from Lemma \ref{lemma:unit-ball-error} using the fact that $\epsilon < q_1(\kappa,A_1,B_1)$ implies that $\Vert  g(A_2,B_2) - g(A_1,B_1) \Vert  < \sigma_{\text{min}}(g(A_1,B_1))$ via the definition of $q_1$ in \eqref{def:q_1}, and \eqref{eqn:lemma-control-error-15} holds due to \eqref{eqn:lemma-control-error-12}.
 
The conclusion follows.
\end{proof}

\begin{lemma} \label{lemma:perturb-bound-matrix-product}
    Consider matrices $M_1,M_2 \in \mathbb{R}^{d_1 \times d_2}$ and $N_1,N_2 \in  \mathbb{R}^{d_2 \times d_3}$. The following result holds:
    \begin{align}
        &\Vert  M_2 N_2 - M_1 N_1 \Vert \leq \Vert M_2 - M_1 \Vert \Vert  N_2 - N_1 \Vert \\
        &\quad + \Vert M_1\Vert \Vert N_2 - N_1 \Vert + \Vert N_1\Vert \Vert M_2 - M_1 \Vert.
    \end{align}
\end{lemma}
\begin{proof}
    The result holds following basic matrix algebra and properties of $\Vert \cdot \Vert$:
    \begin{align}
        &\Vert  M_2 N_2 - M_1 N_1 \Vert  \\
        &= \Vert  M_2 N_2 - M_2 N_1 + M_2 N_1 - M_1 N_1 \Vert \\
        &= \Vert  M_2 (N_2 - N_1) + (M_2 - M_1) N_1 \Vert \\
        &= \Vert  (M_2 - M_1 + M_1) (N_2 - N_1) + (M_2 - M_1) N_1 \Vert \\
        &= \Vert  (M_2 - M_1) (N_2 - N_1) + M_1 (N_2 - N_1) \\
        &\quad + (M_2 - M_1) N_1 \Vert \\
        &\leq \Vert  M_2 - M_1 \Vert \Vert N_2 - N_1\Vert + \Vert M_1\Vert \Vert N_2 - N_1\Vert \\
        &\quad + \Vert M_2 - M_1\Vert \Vert N_1 \Vert.
    \end{align}
\end{proof}
\begin{lemma} \label{lemma:perturb-bound-matrix-inverse}
    Consider nonsingular matrix $M_1 \in \mathbb{R}^{d_1 \times d_1}$. Fix $\delta \in [0, \sigma_{\text{min}}(M_1))$. We have
    \begin{align}
        \Vert  M_2^{-1} - M_1^{-1} \Vert \leq \frac{\Vert  M_1^{-1} \Vert^2 \delta}{1 - \Vert M_1^{-1}\Vert \delta}
    \end{align}
    for all $M_2 \in \bar{\mathcal{B}}_{\delta}(M_1)$.
\end{lemma}
\begin{proof}
    Suppose $M_2 \in \bar{\mathcal{B}}_{\delta}(M_1)$, such that $M_2$ is nonsingular. Then, using basic matrix algebra and properties of $\Vert \cdot \Vert$, we have:
    \begin{align}
        &\Vert  M_2^{-1} - M_1^{-1} \Vert \\
        &= \Vert  M_2^{-1}(M_1 - M_2) M_1^{-1} \Vert \\
        &= \Vert  (M_2^{-1}-M_1^{-1} + M_1^{-1} )(M_1 - M_2) M_1^{-1} \Vert \\
        &\leq \Vert  M_2^{-1}-M_1^{-1} + M_1^{-1} \Vert \Vert  M_1 - M_2 \Vert \Vert  M_1^{-1}\Vert \\
        &\leq \Vert  M_2^{-1}-M_1^{-1} \Vert \Vert  M_1^{-1}\Vert \Vert  M_1 - M_2 \Vert  \\
        &\quad + \Vert  M_1^{-1} \Vert^2 \Vert  M_1 - M_2 \Vert\\
        &\leq \Vert  M_2^{-1}-M_1^{-1} \Vert \Vert  M_1^{-1}\Vert \delta  + \Vert  M_1^{-1} \Vert^2 \delta. \label{eqn:lemma-perturb-multi-1}
    \end{align}
    Rearranging \eqref{eqn:lemma-perturb-multi-1}, we conclude
    \begin{align}
        &\Vert  M_2^{-1} - M_1^{-1} \Vert - \Vert  M_2^{-1}-M_1^{-1} \Vert \Vert  M_1^{-1}\Vert \delta  \leq \Vert  M_1^{-1} \Vert^2 \delta \\
        &\Vert  M_2^{-1} - M_1^{-1} \Vert(1 - \Vert  M_1^{-1}\Vert \delta ) \leq \Vert  M_1^{-1} \Vert^2 \delta \\
        &\Vert  M_2^{-1} - M_1^{-1} \Vert \leq  \frac{\Vert  M_1^{-1} \Vert^2 \delta}{1 - \Vert  M_1^{-1}\Vert \delta}, \label{eqn:lemma-perturb-multi-2}
    \end{align}
    where \eqref{eqn:lemma-perturb-multi-2} holds since $\delta \in [0,\sigma_{\text{min}}(M_1))$.
\end{proof}

\begin{lemma} \label{lemma:perturb-bound-matrix-power}
    Consider $M_1 \in \mathbb{R}^{d \times d}$. We have
    \begin{align}
        \Vert M_2^i - M_1^i\Vert \leq q_{10}(\delta,i,M_1) \label{eqn:lemma-perturb-bound-matrix-power-claim}
    \end{align}
    for all $\delta > 0$, $i \in \mathbb{N}$ and $M_2 \in \bar{\mathcal{B}}_{\delta}(M_1)$,
    where
    \begin{align}
        q_{10}(\delta,i,M_1) := \begin{cases}
            \delta, \quad i = 1 \\
            q_{10}(\delta,i-1,M_1) \delta + \Vert M_1^{i-1} \Vert \delta \\
            \quad + \Vert M_1\Vert q_{10}(\delta,i-1,M_1), \quad i \geq 2.
        \end{cases} \label{def:q_10}
    \end{align} 
\end{lemma}

\begin{proof}
    Suppose $\delta$, and $M_2 \in \bar{\mathcal{B}}_{\delta}(M_1)$. We proceed via proof by induction. When $i = 1$, we have that $\Vert M_2^i - M_1^i\Vert \leq \delta$ since $M_2 \in \bar{\mathcal{B}}_{\delta}(M_1)$. Now, assume that the claim \eqref{eqn:lemma-perturb-bound-matrix-power-claim} holds for $i = n \in \mathbb{N}$. We have
    \begin{align}
        &\Vert M_2^{n+1} - M_1^{n+1} \Vert \\
        &\leq \Vert M_2^n - M_1^n\Vert \Vert  M_2 - M_1 \Vert  + \Vert M_1^n\Vert \Vert  M_2 - M_1 \Vert \\
        &\quad + \Vert M_1\Vert \Vert M_2^n - M_1^n \Vert \\
        &\leq \Vert M_2^n - M_1^n\Vert \delta + \Vert M_1^n\Vert \delta + \Vert M_1\Vert \Vert M_2^n - M_1^n \Vert \\
        &\leq q_{10}(\delta,n,M_1) \delta + \Vert M_1^n\Vert \delta + \Vert M_1\Vert q_{10}(\delta,n,M_1), \label{eqn:lemma-perturb-bound-matrix-power-1}
    \end{align}
    where \eqref{eqn:lemma-perturb-bound-matrix-power-1} follows from the assumption that \eqref{eqn:lemma-perturb-bound-matrix-power-claim} holds. The conclusion follows.
\end{proof}

\begin{lemma} \label{lemma:perturb-matrix-power-product}
    Consider $M_1 \in \mathbb{R}^{d_1 \times d_1}$, and $N_1 \in \mathbb{R}^{d_1 \times d_2}$. We have
    \begin{align}
        \Vert M_2^i N_2 - M_1^i N_1\Vert \leq q_9(\delta,i,M_1,N_1)
    \end{align}
    for all $\delta > 0$, $i \in \mathbb{N}_0$, $M_2 \in \bar{\mathcal{B}}_{\delta}(M_1)$, and $N_2 \in \bar{\mathcal{B}}_{\delta}(N_1)$, where
    \begin{align}
        q_9(\delta,i,M_1,N_1) := \begin{cases}
            \delta, \quad i = 0 \\
            q_{10}(\delta,i,M_1) \delta + \Vert M_1^i\Vert \delta \\
            \quad + \Vert N_1\Vert q_{10}(\delta,i,M_1), \quad i \geq 1
        \end{cases} \label{def:q_9}
    \end{align} 
\end{lemma}

\begin{proof}
    Suppose $\delta > 0$, $M_2 \in \bar{\mathcal{B}}_{\delta}(M_1)$ and $N_2 \in \bar{\mathcal{B}}_{\delta}(N_1)$. We proceed by considering the case where $i=0$ and $i \in \mathbb{N}$ separately.
    
    \textit{Case 1:} Suppose $i = 0$. Then, $\Vert M_2^iN_2 - M_1^iN_1\Vert = \Vert N_2 - N_1\Vert \leq \delta$ holds. 
    
    \textit{Case 2:} Suppose $i \in \mathbb{N}$. Then,
    \begin{align}
        &\Vert M_2^iN_2 - M_1^iN_1\Vert \\
        &\leq \Vert M_2^i - M_1^i \Vert \Vert  N_2 - N_1 \Vert + \Vert M_1^i\Vert \Vert N_2 - N_1 \Vert \\
        &\quad + \Vert N_1\Vert \Vert M_2^i - M_1^i \Vert \label{eqn:lemma-perturb-matrix-power-product-1}\\
        &\leq \Vert M_2^i - M_1^i \Vert \delta + \Vert M_1^i\Vert \delta + \Vert N_1\Vert \Vert M_2^i - M_1^i \Vert \\
        &\leq q_{10}(\delta,i,M_1) \delta + \Vert M_1^i\Vert \delta + \Vert N_1\Vert q_{10}(\delta,i,M_1), \label{eqn:lemma-perturb-matrix-power-product-2}
    \end{align}
    where \eqref{eqn:lemma-perturb-matrix-power-product-1} holds via Lemma \ref{lemma:perturb-bound-matrix-product}, and \eqref{eqn:lemma-perturb-matrix-power-product-2} holds via Lemma \ref{lemma:perturb-bound-matrix-power}. The conclusion follows.
\end{proof}

\begin{lemma} \label{lemma:sat-func-bound}
    Consider $r > 0$. Then,
    \begin{align}
        &|\sat_r(x_2) - \sat_r(x_1)| \\
        &\leq \begin{cases} 
            r | \frac{x_2}{|x_2|} - \frac{x_1}{|x_1|} |, \quad |x_1| > r, \ |x_2| > r\\
            |x_2 - x_1 |, \quad x_1 \in \mathbb{R}^d, \ x_2 \in \mathbb{R}^d,
        \end{cases}
    \end{align}
    holds for all $x_1 \in \mathbb{R}^d$, and $x_2 \in \mathbb{R}^2$.
\end{lemma}

\begin{proof}
    Consider the case where $x_1 \in \mathbb{R}^d$, and $x_2 \in \mathbb{R}^d$. Then,
    \begin{align}
        &|\sat_r(x_2) - \sat_r(x_1) | \\
        &= | \proj_{\bar{\mathcal{B}}_r(0)}(x_2) - \proj_{\bar{\mathcal{B}}_r(0)}(x_1) | \label{eqn:lemma-sat-func-bound-1} \\
        &\leq |x_2 - x_1|, \label{eqn:lemma-sat-func-bound-2}
    \end{align}
    where \eqref{eqn:lemma-sat-func-bound-1} follows from the fact that $\sat_r(x)=\proj_{\bar{\mathcal{B}}_r(0)}(x)$ for all $x \in \mathbb{R}^d$, and \eqref{eqn:lemma-sat-func-bound-2} follows from the nonexpansive property of projection onto a closed convex set (see \cite[Proposition~B.11]{bertsekas1997nonlinear}).

    Now, consider the case where $|x_1| > r$ and $|x_2| > r$. We have
    \begin{align}
        |\sat_r(x_2) - \sat_r(x_1) | &= |r \frac{x_2}{|x_2|} - r \frac{x_1}{|x_1|} | \\
        &= r| \frac{x_2}{|x_2|} - \frac{x_1}{|x_1|} |.
    \end{align}

    The conclusion follows.
\end{proof}

\begin{lemma} \label{lemma:unit-ball-error}
    Consider $M_1 \in \mathbb{R}^{d_1 \times d_2}$ with with $\sigma_{\text{min}}(M_1) > 0$. Then,
    \begin{align}
        \Big| \frac{M_2 x}{|M_2 x|} - \frac{M_1 x}{|M_1 x|} \Big| &\leq \frac{2\delta}{\sigma_{\text{min}}(M_1) - \delta}
    \end{align}
    for all $\delta \in [0,\sigma_{\text{min}}(M_1))$, $M_2 \in \bar{\mathcal{B}}_{\delta}(M_1)$, and $x \in \mathbb{R}^{d_2}$.
\end{lemma}

\begin{proof}
    Suppose $\delta \in [0,\sigma_{\text{min}}(M_1))$, $M_2 \in \bar{\mathcal{B}}_{\delta}(M_1)$, and $x \in \mathbb{R}^{d_2}$. Using basic matrix algebra and properties of $\Vert \cdot \Vert$, we have
    \begin{align}
        &\Big| \frac{M_2 x}{|M_2 x|} - \frac{M_1 x}{|M_1 x|} \Big| \\
        &= \Big| \frac{(M_1 + M_2 - M_1) x}{|(M_1 + M_2 - M_1) x|} - \frac{M_1 x}{|M_1 x|} \Big| \\
        &= \Big| \frac{M_1 x}{|(M_1 + M_2 - M_1) x|} \\
        &\quad + \frac{(M_2 - M_1) x}{|(M_1 + M_2 - M_1) x|} - \frac{M_1 x}{|M_1 x|} \Big| \\
        &\leq \Big| \frac{(M_2 - M_1) x}{|(M_1 + M_2 - M_1) x|} \Big| \\
        &\quad + \Big| \frac{M_1 x}{|(M_1 + M_2 - M_1) x|} - \frac{M_1 x}{|M_1 x|} \Big| \label{eqn:lemma-unit-ball-5}
    \end{align}

    Next, note that the following holds:
    \begin{align}
        \sigma_{\text{min}}(M_1 + M_2 - M_1) &\geq \min_{|x|=1} | |M_1 x| - |(M_2 - M_1)x| | \\
        & \geq \sigma_{\text{min}}(M_1) - \Vert M_2 - M_1 \Vert. \label{eqn:lemma-unit-ball-2}
    \end{align}
    
    The left hand side of \eqref{eqn:lemma-unit-ball-5} can be upper bounded as follows:
    \begin{align}
        \Big| \frac{(M_2 - M_1) x}{|(M_1 + M_2 - M_1) x|} \Big| &\leq \frac{\Vert  M_2 - M_1 \Vert |x|}{\sigma_{\text{min}} (M_1 + M_2 - M_1)|x|} \\
        &= \frac{\Vert  M_2 - M_1 \Vert }{\sigma_{\text{min}} (M_1 + M_2 - M_1)} \\
        &\leq \frac{\Vert  M_2 - M_1 \Vert }{\sigma_{\text{min}} (M_1) - \Vert  M_2 - M_1 \Vert} \label{eqn:lemma-unit-ball-1}\\
        &\leq \frac{\delta }{\sigma_{\text{min}} (M_1) - \delta}, \label{eqn:lemma-unit-ball-4}
    \end{align}
    where \eqref{eqn:lemma-unit-ball-1} follows from \eqref{eqn:lemma-unit-ball-2}.

    The right hand side of \eqref{eqn:lemma-unit-ball-5} is upper bounded as follows:
    \begin{align}
        &\Big| \frac{M_1 x}{|(M_1 + M_2 - M_1) x|} - \frac{M_1 x}{|M_1 x|} \Big| \\
        &= \Big| M_1 x \Big( \frac{1}{|(M_1 + M_2 - M_1)x|} - \frac{1}{| M_1 x |} \Big) \Big| \\
        &\leq |M_1 x| \Big| \frac{1}{|(M_1 + M_2 - M_1)x|} - \frac{1}{| M_1 x |}  \Big| \\
        &= |M_1 x| \Big| \frac{ |M_1 x| - |(M_1 + M_2 - M_1) x| }{|M_1 x | | (M_1 + M_2 - M_1) x |} \Big| \\
        &= \Big| \frac{ |M_1 x| - |(M_1 + M_2 - M_1) x| }{ | (M_1 + M_2 - M_1) x |} \Big| \\
        &\leq  \frac{ | |M_1 x| - |(M_1 + M_2 - M_1) x| | }{  \sigma_{\text{min}}(M_1 + M_2 - M_1) | x |} \\
        &\leq \frac{ | M_1 x - (M_1 + M_2 - M_1) x|  }{  \sigma_{\text{min}}(M_1 + M_2 - M_1) | x |} \\
        &= \frac{ | (M_1 - M_2) x|  }{  \sigma_{\text{min}}(M_1 + M_2 - M_1) | x |} \\
        &\leq \frac{\delta}{\sigma_{\text{min}}(M_1) - \delta} \label{eqn:lemma-unit-ball-3}
    \end{align}
    where \eqref{eqn:lemma-unit-ball-3} follows from \eqref{eqn:lemma-unit-ball-4}.
    
    The conclusion follows by combining \eqref{eqn:lemma-unit-ball-5} with \eqref{eqn:lemma-unit-ball-4} and \eqref{eqn:lemma-unit-ball-3}.
\end{proof}

\subsection{Proof of Supporting Results for Estimation Error Bound}

\begin{proof}[Proof of Proposition \ref{prop:estim-bound-all-time}]
    From the definition of our control strategy in \eqref{eqn:control-strategy}, we have
    \begin{align}
        \mathbb{E}[|U_t|^2] &\leq \mathbb{E}[(D + |V_t|)^2] = 2(D^2 + \trace{\Sigma_V}) \label{eqn:prop-estim-bound-all-time-5}
    \end{align}
    for all $t \in \mathbb{N}_0$, where \eqref{eqn:prop-estim-bound-all-time-5} follows from the AM-QM inequality.

    Moreover, for all $t \in \mathbb{N}_0$, we have
    \begin{align}
        &\mathbb{E}[|X_t|^2] \\
        &= \mathbb{E}[| A^t x_0 + \sum_{s=0}^{t-1} A^s B U_{t-1-s} + \sum_{s=0}^{t-1} A^s W_{t-1-s}|^2] \label{eqn:prop-estim-bound-all-time-1} \\
        &\leq \mathbb{E}[(\Vert  A \Vert^t | x_0 | + \sum_{s=0}^{t-1} \Vert  A \Vert^s \Vert B\Vert |U_{t-1-s}| \\
        &\quad + \sum_{s=0}^{t-1} \Vert  A \Vert^s |W_{t-1-s}|)^2] \\
        &\leq \mathbb{E}[(| x_0 | + t \Vert B\Vert D + \Vert B\Vert \sum_{s=0}^{t-1} |V_{t-1-s}| \\
        &\quad + \sum_{s=0}^{t-1} |W_{t-1-s}|)^2] \label{eqn:prop-estim-bound-all-time-2} \\
        &\leq \mathbb{E}[4 ( | x_0 |^2 + t^2 \Vert B\Vert^2 D^2 \\
        &\quad + t\Vert B\Vert^2 \sum_{s=0}^{t-1} |V_{t-1-s}|^2 + t\sum_{s=0}^{t-1} |W_{t-1-s}|^2 )] \label{eqn:prop-estim-bound-all-time-3} \\
        &= 4 ( | x_0 |^2 + t^2 (\Vert B\Vert^2 (D^2 + \trace{\Sigma_V}) +  \trace{\Sigma_W} )), \label{eqn:prop-estim-bound-all-time-4}
    \end{align}
    where \eqref{eqn:prop-estim-bound-all-time-1} follows via iterative application of \eqref{eqn:open-system}, \eqref{eqn:prop-estim-bound-all-time-2} follows from \eqref{eqn:prop-estim-bound-all-time-5} and Assumption \ref{assump:reachable}, \eqref{eqn:prop-estim-bound-all-time-3} follows via the AM-QM inequality, and \eqref{eqn:prop-estim-bound-all-time-4} follows via linearity of expectation and $\mathbb{E}[|V_s|^2] = \tr(\Sigma_V)$, $\mathbb{E}[|W_s|^2] = \tr(\Sigma_W)$ for all $s \in \mathbb{N}_0$
    
    Next, for all $T \in \mathbb{N}$, we have
    \begin{align}
        &\sum_{t=1}^T \trace{\mathbb{E}[Z_tZ_t^{\top}]} = \sum_{t=1}^T \mathbb{E}[|X_{t-1}|^2] + \sum_{t=1}^T \mathbb{E}[|U_{t-1}|^2] \\
        & \leq 4 \sum_{t=1}^T ( | x_0 |^2 + (t-1)^2 (\Vert B\Vert^2 ( D^2 + \trace{\Sigma_V}) +  \trace{\Sigma_W} ))  \\
        & \quad + 2 \sum_{t=1}^T (D^2 + \trace{\Sigma_V}) \label{eqn:prop-estim-bound-all-time-6} \\
        & \leq 4 T ( | x_0 |^2 + T^2 (\Vert B\Vert^2 (D^2 + \trace{\Sigma_V}) +  \trace{\Sigma_W} ))  \\
        & \quad + 2 T (D^2 + \trace{\Sigma_V})  \\
        &= T ( 4 |x_0|^2 + 2(D^2 + \trace{\Sigma_V})  \\
        &\quad + 4(\Vert B\Vert^2 (D^2 + \trace{\Sigma_V})+  \trace{\Sigma_W} ) T^2 ), \label{eqn:prop-estim-bound-all-time-7}
    \end{align}
    where \eqref{eqn:prop-estim-bound-all-time-6} follows from \eqref{eqn:prop-estim-bound-all-time-5} and \eqref{eqn:prop-estim-bound-all-time-4}, and \eqref{eqn:prop-estim-bound-all-time-7} follows from $t-1 \leq T$ for all $t \in \{ 1, \hdots, T\}$.
    
    We now derive a high probability positive semidefinite upper bound on $\sum_{t=1}^T Z_tZ_t^{\top}$. For convenience, let $C_1(T,x_0)=( 4 |x_0|^2 + 2(D^2 + \trace{\Sigma_V}) + 4(\Vert B\Vert^2 (D^2 + \trace{\Sigma_V}) +  \trace{\Sigma_W} ) T^2  + \lambda_{\text{max}}(\Gamma_{\text{sb}}))$. For all $T \in \mathbb{N}$ and $\delta' \in (0,1)$,
    \begin{align}
        &P( \sum_{t=1}^T Z_tZ_t^{\top} \not \preceq \frac{T}{\delta'} C_1(T,x_0) I ) \\
        & = P( (T C_1(T,x_0) )^{-1} \lambda_{\text{max}} (\sum_{t=1}^T Z_t Z_t^{\top}) \geq \frac{1}{\delta'} ) \label{eqn:prop-estim-bound-all-time-8} \\
        & \leq \delta' \mathbb{E} [ (T C_1(T,x_0) )^{-1} \lambda_{\text{max}} (\sum_{t=1}^T Z_t Z_t^{\top}) ] \label{eqn:prop-estim-bound-all-time-9} \\
        & \leq \delta' (T C_1(T,x_0) )^{-1} \mathbb{E} [ \tr(\sum_{t=1}^T Z_t Z_t^{\top}]) \label{eqn:prop-estim-bound-all-time-10} \\
        & =  \delta' (T C_1(T,x_0) )^{-1}  \sum_{t=1}^T \tr( \mathbb{E} [ Z_t Z_t^{\top}] ) \leq \delta', \label{eqn:prop-estim-bound-all-time-11}
    \end{align} 
    where \eqref{eqn:prop-estim-bound-all-time-8} follows from the definition of $\preceq$, \eqref{eqn:prop-estim-bound-all-time-9} follows from Markov's inequality, \eqref{eqn:prop-estim-bound-all-time-10} holds since $\lambda_{\text{max}}(\cdot) \leq \tr(\cdot)$, and \eqref{eqn:prop-estim-bound-all-time-11} follows from \eqref{eqn:prop-estim-bound-all-time-7}.
        
    For ease of notation over the remainder of the proof, let $p(T,\delta',x_0)=\frac{90 \sqrt{\lambda_{\text{min}}(\Sigma_W)}}{p}  ( (T \lambda_{\text{min}}(\Gamma_{\text{sb}}) )^{-1} ( n + (n+m) \ln (10/p) + \ln \det ( (1/\delta')C_1(T,x_0) \Gamma_{\text{sb}}^{-1} ) + \ln (1/\delta') )  )^{1/2}$. We now derive a finite sample estimation error bound for the least squares parameter estimate obtained from \eqref{eqn:parameter-estimate} by satisfying the premise of Proposition \ref{prop:estim-bound}. In particular, suppose $\delta' \in (0,1)$ and $T \in \mathbb{N}$, and consider the covariate-response sequence $(Z_t,X_t)_{t \in \mathbb{N}}$. Then, (a) $X_t = \theta_* Z_t + W_{t-1}$ for $t \leq T$ holds, where $W_{t-1} \mid \mathcal{F}_{t-1}$ is mean-zero and $\lambda_{\text{max}}(\Sigma_W)$-sub-Gaussian with $\mathcal{F}_t$ denoting the sigma-algebra generated by $W_0,\hdots,W_{t-1},Z_1,\hdots,Z_t$. Moreover, (b) $(Z_1,\hdots,Z_T)$ satisfies the $(k,\Gamma_{\text{sb}},p)$-BMSB condition as per Assumption \ref{assump:bmsb}. Finally, (c) $P( \sum_{t=1}^T Z_tZ_t^{\top} \not \preceq \frac{T}{\delta'} C_1(T,x_0) I ) \leq \delta'$ is satisfied from \eqref{eqn:prop-estim-bound-all-time-11}. Thus, using Proposition \ref{prop:estim-bound} it follows that for all $\delta' \in (0,1)$ and $T \in \mathbb{N}$, if $T \geq \frac{10 k}{p^2}(\ln(1/\delta')+2(n+m)\ln(10/p)+\ln \det (\frac{1}{\delta'} C_1(T,x_0) \Gamma_{\text{sb}}^{-1}))$, we have
    \begin{align}
        &P(\Vert \hat{\theta}_T-\theta_*\Vert _{2} > p(T,\delta',x_0)) \leq 3 \delta'. \label{eqn:prop-estim-bound-all-time-12}
    \end{align}
    
    Let $C_2=1/(-1+\pi^2/6)$. We conclude that
    \begin{align}
        &P(\Vert \hat{\theta}_T-\theta_*\Vert _{2} \leq e(T,\delta,x_0) \text{ for all } T \geq T_0(\delta,x_0)) \\
        & = P( \bigcap_{T \geq T_0(\delta,x_0)} \{ \Vert \hat{\theta}_T-\theta_*\Vert _{2} \leq e(T,\delta,x_0) \}) \\
        &= 1 - P( \bigcup_{T \geq T_0(\delta,x_0)} \{ \Vert \hat{\theta}_T-\theta_*\Vert _{2} > p(T,\frac{C_2\delta}{3(T+1)^2}, x_0) \}) \label{eqn:prop-estim-bound-all-time-16} \\
        & \geq 1 - \sum_{T \geq T_0(\delta,x_0)}P( \{ \Vert \hat{\theta}_T-\theta_*\Vert _{2} > p(T,\frac{C_2\delta}{3(T+1)^2}, x_0) \}) \label{eqn:prop-estim-bound-all-time-13} \\
        & \geq 1 - C_2 \sum_{T \geq T_0(\delta,x_0)} \frac{\delta}{(T+1)^2} \label{eqn:prop-estim-bound-all-time-14} \\
        & \geq 1 - C_2 \sum_{T \geq 1} \frac{\delta}{(T+1)^2} = 1 - \delta, \label{eqn:prop-estim-bound-all-time-17}
     \end{align}
     where \eqref{eqn:prop-estim-bound-all-time-16} holds since $e(T,\delta,x_0)=p(T,C_2 \delta/(3 (T+1)^2)$, \eqref{eqn:prop-estim-bound-all-time-13} follows via the union bound, \eqref{eqn:prop-estim-bound-all-time-14} follows from \eqref{eqn:prop-estim-bound-all-time-12}, and \eqref{eqn:prop-estim-bound-all-time-17} follows since $T_0(\delta,x_0) \geq 1$, and $\sum_{T\geq 1} 1/(T+1)^2 = -1 + \pi^2/6=1/C_2$.
\end{proof}

\subsection{Proof of Supporting Results for Stability Bound} 

\begin{proof}[Proof of Lemma \ref{lemma:stability-deterministic-param}]
    Let $\bar{\mathcal{R}} = \mathcal{R}_{\kappa}(\bar{A},\bar{B})$, and $K = \{ z \in \mathbb{R}^n \mid |g(A,B)z| \leq D \}$. Suppose $\bar{A} \in \bar{\mathcal{B}}_{\epsilon}(A)$ and $\bar{B} \in \bar{\mathcal{B}}_{\epsilon}(B)$, and consider the corresponding random sequence $(Z^{\bar{A},\bar{B}}_{\tau})_{\tau \in \mathbb{N}_0}$. Since $(Z^{\bar{A},\bar{B}}_{\tau})_{\tau \in \mathbb{N}_0}$ is a Markov process, our proof proceeds by satisfying the conditions in Proposition \ref{prop:geometric-stability} with $V(\cdot) \leftarrow e^{|\cdot|}$, $K$, $\lambda \leftarrow \lambda(\epsilon)$ and $\beta \leftarrow \beta(\epsilon)$.
    
    Firstly, note that the following holds:
        \begin{align}
        &\mathbb{E}[e^{|Z^{\bar{A},\bar{B}}_1|} \mid Z_0 = z] \\
        &= \mathbb{E}[e^{| A^{\kappa}Z^{\bar{A},\bar{B}}_0 + \mathcal{R}_*\text{sat}_D(-\bar{\mathcal{R}}^{\dagger} \bar{A}^{\kappa} Z^{\bar{A},\bar{B}}_0) + \tilde{V}_0 + \tilde{W}_0 |} \mid Z^{\bar{A},\bar{B}}_0 = z] \\
        &= \mathbb{E}[e^{| A^{\kappa}z + \mathcal{R}_*\text{sat}_D(-\bar{\mathcal{R}}^{\dagger} \bar{A}^{\kappa} z) + \tilde{V}_0 + \tilde{W}_0 |} ] \\
        &\leq e^{| A^{\kappa}z + \mathcal{R}_*\text{sat}_D(-\bar{\mathcal{R}}^{\dagger} \bar{A}^{\kappa} z)| + M_{\bar{V}} + M_{\bar{W}} } \\
        &\leq e^{|\mathcal{R}_*\text{sat}_D(-\bar{\mathcal{R}}^{\dagger} \bar{A}^{\kappa} z) - \mathcal{R}_*\text{sat}_D(-\mathcal{R}_*^{\dagger} A^{\kappa} z)| + M_{\bar{V}} + M_{\bar{W}} } \\
        &\quad \times e^{| A^{\kappa}z + \mathcal{R}_*\text{sat}_D(-\mathcal{R}_*^{\dagger} A^{\kappa} z)|}. \label{eqn:lemma-stability-deterministic-param-1}
    \end{align}
    
    We now show that \eqref{eqn:geometric-condition-1} is satisfied in Proposition \ref{prop:geometric-stability}. 
    Suppose $z \not \in K$. Since $|\mathcal{R}_*^{\dagger} A^{\kappa} z| > D$ is satisfied, we have
    \begin{align}
        &| A^{\kappa}z + \mathcal{R}_*\text{sat}_D(-\mathcal{R}_*^{\dagger} A^{\kappa} z)| - |z| \\
        &=| A^{\kappa}z - \frac{\mathcal{R}_* \mathcal{R}_*^{\dagger}A^{\kappa}zD}{| \mathcal{R}_*^{\dagger} A^{\kappa} z |} | - |z| \label{eqn:lemma-stability-deterministic-param-2} \\
        &= | \frac{A^{\kappa} z}{ | A^{\kappa} z | } ( |A^{\kappa}z| - \frac{D|A^{\kappa}z|}{| \mathcal{R}_*^{\dagger} A^{\kappa} z |} ) | - |z| \\
        &= |A^{\kappa}z| - \frac{D|A^{\kappa}z|}{| R^{\dagger} A^{\kappa} z |}  - |z| \\ 
        &\leq \Vert A^{\kappa}\Vert|z| - \frac{D|A^{\kappa}z|}{| R^{\dagger} A^{\kappa} z |}  - |z| \\
        &\leq  - \frac{D|A^{\kappa}z|}{| \mathcal{R}_*^{\dagger} A^{\kappa} z |} \leq - \frac{D}{\Vert  \mathcal{R}_*^{\dagger} \label{eqn:lemma-stability-deterministic-param-3}  \Vert}
    \end{align}
    where \eqref{eqn:lemma-stability-deterministic-param-2} follows from the definition of $\text{sat}_D(\cdot)$, and \eqref{eqn:lemma-stability-deterministic-param-3} follows from Assumption \ref{assump:reachable}, and the fact that $\Vert  \mathcal{R}_*^{\dagger}\Vert \geq \frac{ | \mathcal{R}_*^{\dagger}A^{\kappa}z| }{|A^{\kappa}z|}$.
    Next, we have 
    \begin{align}
        &|\mathcal{R}_*(\text{sat}_D(-\bar{\mathcal{R}}^{\dagger} \bar{A}^{\kappa} z) - \text{sat}_D(-\mathcal{R}_*^{\dagger} A^{\kappa} z) )| \\
        & \leq \Vert \mathcal{R}_* \Vert |\text{sat}_D(-\bar{\mathcal{R}}^{\dagger} \bar{A}^{\kappa} z) - \text{sat}_D(-\mathcal{R}_*^{\dagger} A^{\kappa} z )|\\
        &\leq  \Vert \mathcal{R}_*\Vert M_q D \epsilon \label{eqn:lemma-stability-deterministic-param-4}
    \end{align}
    where \eqref{eqn:lemma-stability-deterministic-param-4} follows from Lemma \ref{lemma:control-error-simplified}. Combining \eqref{eqn:lemma-stability-deterministic-param-3}, \eqref{eqn:lemma-stability-deterministic-param-4} and \eqref{eqn:lemma-stability-deterministic-param-1}, we find
    \begin{align}
        &\frac{\mathbb{E}[e^{| Z^{\bar{A},\bar{B}}_1 |} \mid Z_0 = z]}{e^{ | z | }} \\
        & \leq e^{| A^{\kappa}z + \mathcal{R}_*\text{sat}_D(-\mathcal{R}_*^{\dagger} A^{\kappa} z)| - |z|} \\
        & \quad \times e^{|\mathcal{R}_*(\text{sat}_D(-\bar{\mathcal{R}}^{\dagger} \bar{A}^{\kappa} z) - \text{sat}_D(-\mathcal{R}_*^{\dagger} A^{\kappa} z) )| + M_{\bar{V}} + M_{\bar{W}}} \\ \label{eqn:lemma-stability-deterministic-param-6}
        &\leq e^{- \frac{D}{\Vert  \mathcal{R}_*^{\dagger} \Vert} + \Vert \mathcal{R}_*\Vert M_q D \epsilon + M_{\bar{V}} + M_{\bar{W}}} = \lambda(\epsilon),
    \end{align}
    thus satisfying \eqref{eqn:geometric-condition-1} for all $z \not \in K$. 

    Next, we show that \eqref{eqn:geometric-condition-2} is satisfied in Proposition \ref{prop:geometric-stability}. For all $z \in K$, $|g(A,B)z| \leq D$ holds, so
    \begin{align}
        &| A^{\kappa}z + \mathcal{R}_*\text{sat}_D(-\mathcal{R}_*^{\dagger} A^{\kappa} z)| \\
        &= | A^{\kappa}z - \mathcal{R}_*\mathcal{R}_*^{\dagger} A^{\kappa} z| = 0. \label{eqn:lemma-stability-deterministic-param-5}
    \end{align}
    Combining \eqref{eqn:lemma-stability-deterministic-param-5}, \eqref{eqn:lemma-stability-deterministic-param-4}, and \eqref{eqn:lemma-stability-deterministic-param-1}, we find
    \begin{align}
        &\mathbb{E}[e^{|Z^{\bar{A},\bar{B}}_1|} \mid Z_0 = z] \leq e^{\Vert R\Vert M_q D \epsilon + M_{\bar{V}} + M_{\bar{W}} } = \beta(\epsilon),
    \end{align}
    for all $z \in K$, thus satisfying \eqref{eqn:geometric-condition-2}.

    The conclusion then follows from Proposition \ref{prop:geometric-stability}.
\end{proof}


\begin{proof}[Proof of Lemma \ref{lemma:bound-time-stable}]
    We start by deriving an upper bound on $T_0(\delta,x_0)$.
    Let $K_1=10k/p^2$, $f_1(\delta)=3(-1+\pi^2/6)/\delta$, $K_2=2(n+m)\ln(10/p)$, $f_2(x_0,T) = 4 |x_0|^2 + 2(D^2 + \trace{\Sigma_V}) +4(\Vert B\Vert^2 (D^2 + \trace{\Sigma_V}) +  \trace{\Sigma_W} ) T^2 + \lambda_{\text{max}}(\Gamma_{\text{sb}})$, such that
    \begin{align}
        T_0(\delta,x_0) = \min\{ T_0' \in \mathbb{N} \mid T \geq K_1 ( \ln(f_1(\delta)(T+1)^2) + K_2 + \ln( \det ( f_1(\delta)(T+1)^2 f_2(x_0,T) \Gamma_{\text{sb}}^{-1} ) ) ) \text{ for all $T \geq T_0'$} \}
    \end{align}
    holds from the definition of $T_0$ in \eqref{def:T_0}.
    Next, let $f_3(x_0) = 4 |x_0|^2 + 2(D^2 + \trace{\Sigma_V}) +4(\Vert B\Vert^2 (D^2 + \trace{\Sigma_V}) +  \trace{\Sigma_W} ) + \lambda_{\text{max}}(\Gamma_{\text{sb}})$. Then,
    \begin{align}
        &\ln( \det ( f_1(\delta)(T+1)^2 f_2(x_0,T) \Gamma_{\text{sb}}^{-1} ) ) \\
        &\leq \ln( (f_1(\delta)(T+1)^2 f_2(x_0,T))^{n+m} \det ( \Gamma_{\text{sb}}^{-1} ) ) \\
        & \leq \ln( (f_1(\delta)(T+1)^4 f_3(x_0))^{n+m} \det ( \Gamma_{\text{sb}}^{-1} ) ) \label{eqn:lemma-bound-time-stable-1}\\
        & = (n+m) \ln( f_1(\delta) f_3(x_0) ) + 4 (n+m) \ln(T + 1)  + \ln(\det ( \Gamma_{\text{sb}}^{-1} ) ) \label{eqn:lemma-bound-time-stable-2}
    \end{align}
    holds, where \eqref{eqn:lemma-bound-time-stable-1} follows from $ f_2(x_0,T) \leq f_3(x_0)(T+1)^2 $. We now bound
    \begin{align}
        &K_1 ( \ln(f_1(\delta)(T+1)^2) + K_2 + \ln( \det ( f_1(\delta)(T+1)^2 f_2(x_0,T) \Gamma_{\text{sb}}^{-1} ) ) ) \\
        &\leq K_1 ( \ln(f_1(\delta)(T+1)^2) + K_2 + (n+m) \ln( f_1(\delta) f_3(x_0) ) + 4 (n+m) \ln(T + 1)  + \ln(\det ( \Gamma_{\text{sb}}^{-1} ) )  ) \label{eqn:lemma-bound-time-stable-3} \\
        &= K_1 ( 4 (n+m) + 2 ) \ln(T + 1) + K_1 ( (n+m+1)\ln(f_1(\delta)) + K_2 + (n+m) \ln( f_3(x_0) ) + \ln(\det ( \Gamma_{\text{sb}}^{-1} ) ) ) \\
        &= K_3 \ln(T + 1) + f_4(\delta,x_0) \label{eqn:lemma-bound-time-stable-4}
    \end{align}
    where \eqref{eqn:lemma-bound-time-stable-3} follows from \eqref{eqn:lemma-bound-time-stable-2}, and \eqref{eqn:lemma-bound-time-stable-4} holds by setting  $K_3 = K_1 ( 4 (n+m) + 2 )$, and $f_4(\delta,x_0) = K_1 ( (n+m+1) \ln(f_1(\delta)) + K_2 + (n+m) \ln( f_3(x_0) ) + \ln(\det ( \Gamma_{\text{sb}}^{-1} ) ) )$.
    Next, note that 
    \begin{align}
        K_3 \ln (T + 1) \leq (1/2) T + K_3 \ln( 2 K_3 ) - K_3 + 1/2 \label{eqn:lemma-bound-time-stable-7}
    \end{align}
    holds, where the RHS is the tangent line of $K_3 \ln (T + 1)$ with gradient $1/2$. It follows that if
    \begin{align}
        T \geq 2 K_3 \ln (2 K_3) - 2K_3 + 1 + 2f_4(\delta,x_0) \label{eqn:lemma-bound-time-stable-5}
    \end{align}
    then
    \begin{align}
        T &\geq (1/2) T + K_3 \ln( 2 K_3 ) - K_3 + 1/2 + f_4(\delta,x_0) \label{eqn:lemma-bound-time-stable-6} \\
        &\geq K_1 ( \ln(f_1(\delta)(T+1)^2) + K_2 + \ln( \det ( f_1(\delta)(T+1)^2 f_2(x_0,T) \Gamma_{\text{sb}}^{-1} ) ) ) \label{eqn:lemma-bound-time-stable-9}
    \end{align}
    where \eqref{eqn:lemma-bound-time-stable-6} follows by rearranging \eqref{eqn:lemma-bound-time-stable-5}, and \eqref{eqn:lemma-bound-time-stable-9} follows from \eqref{eqn:lemma-bound-time-stable-7} and \eqref{eqn:lemma-bound-time-stable-4}. Thus, we conclude
    \begin{align}
        T_0(\delta,x_0) \leq 2 K_3 \ln (2 K_3) - 2K_3 + 1 + f_4(\delta,x_0). \label{eqn:lemma-bound-time-stable-16}
    \end{align}

    Next, we derive an upper bound on the time $T$ such that $e(t,\delta,x_0) \leq \epsilon$ holds for all $t \geq T$. Let $K_4 = 90\sqrt{\lambda_{\text{max}}(\Sigma_W)}/p$ and $K_5=n + (n+m) \ln(10/p)$, such that
    \begin{align}
        e(T,\delta,x_0) = K_4 \sqrt{ \frac{K_5 + \ln( \det ( f_1(\delta)(T+1)^2 f_2(x_0,T) \Gamma_{\text{sb}}^{-1} ) ) + \ln(f_1(\delta)(T+1)^2) }{ T \lambda_{\text{min}}(\Gamma_{\text{sb}}) } }
    \end{align}
    holds. The inequality $e(T,\delta,x_0) \leq \epsilon$ can then be rearranged as follows:
    \begin{align}
        \epsilon &\geq e(T,\delta,x_0) \\
        \epsilon &\geq K_4 \sqrt{ \frac{K_5 + \ln( \det ( f_1(\delta)(T+1)^2 f_2(x_0,T) \Gamma_{\text{sb}}^{-1} ) ) + \ln(f_1(\delta)(T+1)^2) }{ T \lambda_{\text{min}}(\Gamma_{\text{sb}}) } }  \\
        \frac{\epsilon^2  \lambda_{\text{min}}(\Gamma_{\text{sb}})}{K_4^2} T &\geq K_5 + \ln( \det ( f_1(\delta)(T+1)^2 f_2(x_0,T) \Gamma_{\text{sb}}^{-1} ) ) + \ln(f_1(\delta)(T+1)^2) \\
        T &\geq \frac{K_4^2}{\epsilon^2  \lambda_{\text{min}}(\Gamma_{\text{sb}})} ( K_5 + \ln( \det ( f_1(\delta)(T+1)^2 f_2(x_0,T) \Gamma_{\text{sb}}^{-1} ) ) + \ln(f_1(\delta)(T+1)^2) ). \label{eqn:lemma-bound-time-stable-10}
    \end{align}
    Focusing on the RHS of the inequality, and applying similar steps to \eqref{eqn:lemma-bound-time-stable-4}, we have
    \begin{align}
        &\frac{K_4^2}{\epsilon^2  \lambda_{\text{min}}(\Gamma_{\text{sb}})} ( K_5 + \ln( \det ( f_1(\delta)(T+1)^2 f_2(x_0,T) \Gamma_{\text{sb}}^{-1} ) ) + \ln(f_1(\delta)(T+1)^2) ) \\
        &\leq \frac{K_4^2}{\epsilon^2  \lambda_{\text{min}}(\Gamma_{\text{sb}})} ( K_5 + \ln(f_1(\delta)) + 2\ln(T+1) + (n+m) \ln( f_1(\delta) f_3(x_0) ) + 4 (n+m) \ln(T + 1)  + \ln(\det ( \Gamma_{\text{sb}}^{-1} ) ) ) \\
        &= \frac{K_4^2}{\epsilon^2  \lambda_{\text{min}}(\Gamma_{\text{sb}})} ( K_5 + (n+m+1)\ln(f_1(\delta)) + (n+m) \ln( f_3(x_0) ) + (4 (n+m) + 2) \ln(T + 1)  + \ln(\det ( \Gamma_{\text{sb}}^{-1} ) ) ) \\
        &= K_6 \ln(T + 1) + f_5(x_0,\delta), \label{eqn:lemma-bound-time-stable-15}
    \end{align}
    where $f_5(x_0,\delta) = \frac{K_4^2}{\epsilon^2  \lambda_{\text{min}}(\Gamma_{\text{sb}})}(K_5 + (n+m+1)\ln(f_1(\delta)) + (n+m) \ln( f_3(x_0) ) + \ln(\det(\Gamma_{\text{sb}}^{-1})))$, and $K_6 = \frac{K_4^2}{\epsilon^2  \lambda_{\text{min}}(\Gamma_{\text{sb}})} (4 (n+m) + 2) $.
    Next, note that 
    \begin{align}
        K_6 \ln (T + 1) \leq (1/2) T + K_6 \ln( 2 K_6 ) - K_6 + 1/2 \label{eqn:lemma-bound-time-stable-14}
    \end{align}
    holds, where the RHS is the tangent line of $K_6 \ln (T + 1)$ with gradient $1/2$. It follows that if
    \begin{align}
        T \geq 2 K_6 \ln (2 K_6) - 2K_6 + 1 + 2 f_5(x_0,\delta) \label{eqn:lemma-bound-time-stable-12}
    \end{align}
    then
    \begin{align}
        T &\geq (1/2) T + K_6 \ln( 2 K_6 ) - K_6 + 1/2+ f_5(x_0,\delta) \label{eqn:lemma-bound-time-stable-11} \\
        &\geq \frac{K_4^2}{\epsilon^2  \lambda_{\text{min}}(\Gamma_{\text{sb}})} ( K_5 + \ln( \det ( f_1(\delta)(T+1)^2 f_2(x_0,T) \Gamma_{\text{sb}}^{-1} ) ) + \ln(f_1(\delta)(T+1)^2) ) \label{eqn:lemma-bound-time-stable-13}
    \end{align}
    where \eqref{eqn:lemma-bound-time-stable-11} holds after rearranging \eqref{eqn:lemma-bound-time-stable-12}, and \eqref{eqn:lemma-bound-time-stable-13} holds from \eqref{eqn:lemma-bound-time-stable-14} and \eqref{eqn:lemma-bound-time-stable-15}, which is equivalent to $e(T,\delta,x_0)\leq \epsilon$. Thus,
    \begin{align}
        \min \{ T \in \mathbb{N} \mid e(t,\delta,x_0) \leq \epsilon \text{ for all } t \geq T \} \leq 2 K_6 \ln (2 K_6) - 2K_6 + 1 + f_5(x_0,\delta) \label{eqn:lemma-bound-time-stable-17}
    \end{align}
    holds.

    Combining \eqref{eqn:lemma-bound-time-stable-16} and \eqref{eqn:lemma-bound-time-stable-17}:
    \begin{align}
        \tau_0'(\epsilon,\delta,x_0) &= \min \{ \tau \in \mathbb{N} \mid \kappa \tau \geq T_0(\delta,x_0), e(\kappa i,\delta, x_0) \leq \epsilon \text{ for all } i \geq \tau \} \\
        &\leq (1/\kappa)\max( 2 K_3 \ln (2 K_3) - 2K_3 + 1 + f_4(\delta,x_0), 2 K_6 \ln (2 K_6) - 2K_6 + 1 + f_5(x_0,\delta) ) \\
        &\leq (1/\kappa) [ \max(2 K_3 \ln (2 K_3) - 2K_3 + 1, 2 K_6 \ln (2 K_6) - 2K_6 + 1) + \max(f_4(x_0,\delta), f_5(x_0,\delta)) ] \label{eqn:lemma-bound-time-stable-18} \\
        &= (1/\kappa)[L_3 + \max(f_4(x_0,\delta), f_5(x_0,\delta))], \label{eqn:lemma-bound-time-stable-19}
    \end{align}
    where \eqref{eqn:lemma-bound-time-stable-18} holds since $\max(a + b, c + d) \leq \max(a,c) + \max(b,d)$, and \eqref{eqn:lemma-bound-time-stable-19} follows after setting $L_3 = \max(2 K_3 \ln (2 K_3) - 2K_3 + 1, 2 K_6 \ln (2 K_6) - 2K_6 + 1)$.
    We then bound $\max(f_4(x_0,\delta), f_5(x_0,\delta)$ as follows:
    \begin{align}
        &\max( f_4(\delta,x_0), f_5(\delta,x_0) ) \\
        &= \max( K_1 [ (n+m+1) [ \ln ( 3(-1 + \pi^2/6) ) + \ln(1/\delta) ] + K_2 + (n+m) \ln( f_3(x_0) ) + \ln(\det ( \Gamma_{\text{sb}}^{-1} ) ) ] , \\
        &\quad \frac{K_4^2}{\epsilon^2  \lambda_{\text{min}}(\Gamma_{\text{sb}})}[K_5 + (n+m+1)( \ln ( 3(-1 + \pi^2/6) ) + \ln(1/\delta) ) + (n+m) \ln( f_3(x_0) ) + \ln(\det(\Gamma_{\text{sb}}^{-1}))] ) \\
        &\leq \max( K_1, \frac{K_4^2}{\epsilon^2  \lambda_{\text{min}}(\Gamma_{\text{sb}})}  ) \big [ (n+m+1)\ln(1/\delta) + (n+m+1) \ln ( 3(-1 + \pi^2/6) ) + (n+m) \ln (f_3(x_0)) + \max(K_5,K_2) \big ] \\
        &= L_4(x_0) + L_5\ln(1/\delta), \label{eqn:lemma-bound-time-stable-20}
    \end{align}
    where $L_4(x_0) = \max( K_1, \frac{K_4^2}{\epsilon^2  \lambda_{\text{min}}(\Gamma_{\text{sb}})}  ) \big[ (n+m+1) \ln ( 3(-1 + \pi^2/6) ) + (n+m) \ln (f_3(x_0)) + \max(K_5,K_2) \big]$, and $L_5 = \max( K_1, \frac{K_4^2}{\epsilon^2  \lambda_{\text{min}}(\Gamma_{\text{sb}})}  ) (n+m+1) $.
    Combining \eqref{eqn:lemma-bound-time-stable-19} with \eqref{eqn:lemma-bound-time-stable-20}, we conclude
    \begin{align}
        \tau_0'(\epsilon,\delta,x_0) &\leq  (1/\kappa) [L_3 + L_4(x_0) + L_5\ln(1/\delta)] \\
        &= L_2(x_0) + L_1 \ln(1/\delta)
    \end{align}
    where $L_2(x_0) = (1/\kappa)(L_3 + L_4(x_0))$, and $L_1 = (1/\kappa)L_5$.
\end{proof}

